\documentclass[a4paper, 11pt]{article}

\usepackage[english]{babel}
\usepackage[utf8]{inputenc}
\usepackage[T1]{fontenc}
\usepackage{enumerate}
\usepackage{enumitem}
\usepackage[nottoc]{tocbibind}
\usepackage{subcaption} 
\usepackage[title,titletoc]{appendix}

\usepackage{etoolbox}
\usepackage{authblk}
\usepackage{chemarrow}

\usepackage[normalem]{ulem}
\usepackage{stmaryrd}
\usepackage[margin=1.05in]{geometry}

\usepackage{amsmath}
\usepackage{dsfont}
\usepackage{shuffle}
\numberwithin{equation}{section}
\usepackage{graphicx}
\usepackage[colorinlistoftodos]{todonotes}
\usepackage[colorlinks=true, allcolors=blue]{hyperref}
\usepackage{bussproofs}
\usepackage{turnstile}
\usepackage{amsthm}
\usepackage{amssymb}
\usepackage{tabularx}
\usepackage[section]{placeins}
\usepackage{xcolor}
\usepackage{multirow}
\usepackage{multicol}
\usepackage[many]{tcolorbox}
\tcbuselibrary{theorems}
\usepackage{array}
\usepackage{colortbl}
\usepackage{placeins}
\usepackage{titlesec}
\usepackage{mdframed}
\usepackage{nicefrac}
\usepackage{listings}
\usepackage{etoolbox}
\usepackage{mathtools}
\usepackage[intoc, english]{nomencl}
\makenomenclature{}
\usepackage{tikz}
\usetikzlibrary{shapes, arrows.meta, positioning}
\usepackage{todonotes}
\usepackage{algorithm}
\usepackage{algpseudocode}
\usepackage{physics,amsmath}
\usepackage[skip=11pt plus5pt, indent=0pt]{parskip}
\usepackage{fancyhdr}
\usepackage{atbegshi}
\usepackage{thmtools,thm-restate}
\usepackage{booktabs}
\usepackage{makecell}
\usepackage{setspace}
\makeatletter
\def\algbackskip{\hskip-\ALG@thistlm}
\makeatother

\makeatletter
\def\thanks#1{\protected@xdef\@thanks{\@thanks
        \protect\footnotetext{#1}}}
\makeatother

\definecolor{dgreen}{RGB}{0,150,0}

\usepackage{tocloft}

\newcommand{\bluebf}[1]{\textcolor{ceruleanblue}{{\bf{#1}}}}

\let\Oldsection\section{}
\renewcommand{\section}{\FloatBarrier\Oldsection}

\let\Oldsubsection\subsection{}
\renewcommand{\subsection}{\FloatBarrier\Oldsubsection}

\let\Oldsubsubsection\subsubsection{}
\renewcommand{\subsubsection}{\FloatBarrier\Oldsubsubsection}
\newtheorem{theorem}{Theorem}[section]
\newtheorem{corollary}{Corollary}[theorem]

\newtheorem{lemma}[theorem]{Lemma}

\theoremstyle{definition}
\newtheorem{definition}[theorem]{Definition}

\theoremstyle{remark}
\newtheorem{remark}{Remark}[section]

\theoremstyle{definition}
\newtheorem{example}{Example}[section]

\newtcbtheorem[use counter*=theorem]{definitionBoxBase}{Definition}{}{def}
\newtcbtheorem[use counter*=theorem]{claimBoxBase}{Claim}{}{def}
\newtcbtheorem[use counter*=theorem]{lemmaBoxBase}{Lemma}{}{lem}
\newtcbtheorem[use counter*=theorem]{figureBoxBase}{Figure}{}{fig}
\newtcbtheorem[use counter*=theorem]{theoremBoxBase}{Theorem}{}{thm}
\newtcbtheorem[use counter*=theorem]{propositionBoxBase}{Proposition}{}{prop}

\hypersetup{citecolor=blue}

\allowdisplaybreaks

\newcommand{\sprod}[2]{\left\langle #1, #2 \right\rangle}

\newcommand{\cX}{\mathcal{X}}
\newcommand{\cH}{\mathcal{H}}
\newcommand{\bR}{\mathbb{R}}
\newcommand{\bfX}{\mathbf{X}}
\newcommand{\bfY}{\mathbf{Y}}
\newcommand{\dD}{\, \mathrm{d}}

\newcommand{\cS}{\mathcal{S}}

\usepackage{amsmath,etoolbox}
\usepackage{commath}
\usepackage{amsthm}
\usepackage{amssymb}
\usepackage{bbm}
\usepackage{amsfonts}
\usepackage{enumitem} 
\AtBeginEnvironment{pmatrix}{\setlength{\arraycolsep}{8pt}}

\definecolor{steelblue}{rgb}{0.27, 0.51, 0.71}
\definecolor{mediumelectricblue}{rgb}{0.01, 0.31, 0.59}
\definecolor{halayaube}{rgb}{0.4, 0.22, 0.33}
\definecolor{eggplant}{rgb}{0.38, 0.25, 0.32}
\definecolor{darkraspberry}{rgb}{0.53, 0.15, 0.34}
 \definecolor{bluepigment}{rgb}{0.2, 0.2, 0.6}
 \definecolor{spirodiscoball}{rgb}{0.06, 0.75, 0.99}
 \definecolor{vividburgundy}{rgb}{0.62, 0.11, 0.21}
 \definecolor{mediumtealblue}{rgb}{0.0, 0.33, 0.71}
 \definecolor{mediumvioletred}{rgb}{0.78, 0.08, 0.52}
 \definecolor{egyptianblue}{rgb}{0.06, 0.2, 0.65}
 \definecolor{dodgerblue}{rgb}{0.12, 0.56, 1.0}
 \definecolor{ceruleanblue}{rgb}{0.16, 0.32, 0.75}
 \definecolor{darktangerine}{rgb}{1.0, 0.66, 0.07}
 \definecolor{charcoal}{rgb}{0.21, 0.27, 0.31}
 \definecolor{persianindigo}{rgb}{0.2, 0.07, 0.48}
 \definecolor{orangepeel}{rgb}{1.0, 0.62, 0.0}
 \definecolor{awesome}{rgb}{1.0, 0.13, 0.32}
 \definecolor{green(ncs)}{rgb}{0.0, 0.62, 0.42}

\definecolor{to-do}{rgb}{1.0, 0.13, 0.32} 
\definecolor{comment}{rgb}{0.01, 0.28, 1.0} 

 \definecolor{nico}{rgb}{1.0, 0.62, 0.0} 
\definecolor{owen}{rgb}{0.01, 0.28, 1.0} 
\definecolor{old}{rgb}{0.43, 0.5, 0.5} 
 \definecolor{blanka}{rgb}{0.0, 0.62, 0.42}

\usepackage{hyperref}
\usepackage{setspace}
\hypersetup{
    colorlinks,
    citecolor=mediumtealblue,
    filecolor=black,
    linkcolor=mediumvioletred,
    urlcolor=darkraspberry 	
}
\usepackage{caption}
\usepackage{algorithm}
\usepackage{algpseudocode}
\usepackage{theoremref}

\DeclareMathAlphabet{\mathpzc}{OT1}{pzc}{m}{it}

\newcommand{\Ex}{\mathbb{E}}

\newcommand{\N}{\mathbb{N}}
\renewcommand{\P}{\mathbb{P}}

\newcommand{\R}{\mathbb{R}}

\newcommand{\bX}{\mathbb{X}}

\newcommand{\F}{\mathcal{F}}

\renewcommand{\H}{\mathcal{H}}
\newcommand{\I}{\mathcal{I}}
\newcommand{\J}{\mathcal{J}}
\newcommand{\K}{\mathcal{K}}

\newcommand{\X}{\mathcal{X}}

\newcommand{\V}{\mathcal{V}}
\newcommand{\W}{\mathcal{W}}

\newcommand{\kappasig}{\kappa_{\textup{sig}}}
\newcommand{\ExP}{\Ex_\P}

\DeclareMathOperator*{\argmax}{argmax}
\DeclareMathOperator*{\argmin}{argmin}

\let\norm\relax
\DeclarePairedDelimiter{\norm}{\lVert}{\rVert}
\DeclarePairedDelimiterX{\inp}[2]{\langle}{\rangle}{#1, #2}
\usepackage[
backend=biber,
style=alphabetic,
sorting=nyt
]{biblatex}
\usepackage{csquotes}
\definecolor{gray75}{gray}{0.75}
\titleformat{\chapter}[hang]{\Huge\bfseries}{\thechapter\hspace{20pt}\textcolor{gray75}{|}\hspace{20pt}}{0pt}{\Huge\bfseries}

\titlespacing*{\chapter}      {0pt}{0pt}{25pt}

\makeatletter
\def\ps@myheadings{%
    \let\@oddfoot\@empty\let\@evenfoot\@empty
    \def\@evenhead{\thepage\hfil\slshape\leftmark}%
    \def\@oddhead{{\slshape\rightmark}\hfil\thepage}%
    \let\@mkboth\@gobbletwo
    \let\sectionmark\@gobble
    \let\subsectionmark\@gobble
    }
  \if@titlepage
  \renewcommand\maketitle{\begin{titlepage}%
  \let\footnotesize\small
  \let\footnoterule\relax
  \let \footnote \thanks
  \null\vfil
  \vskip 60\p@
  \begin{center}%
    {\LARGE \@title \par}%
    \vskip 3em%
    {\large
     \lineskip .75em%
      \begin{tabular}[t]{c}%
        \@author
      \end{tabular}\par}%
      \vskip 1.5em%
    {\large \@date \par}
  \end{center}\par
  \@thanks
  \vfil\null
  \end{titlepage}%
  \setcounter{footnote}{0}%
}
\else
\renewcommand\maketitle{\par
  \begingroup
    \renewcommand\thefootnote{\@fnsymbol\c@footnote}%
    \def\@makefnmark{\rlap{\@textsuperscript{\normalfont\@thefnmark}}}%
    \long\def\@makefntext##1{\parindent 1em\noindent
            \hb@xt@1.8em{%
                \hss\@textsuperscript{\normalfont\@thefnmark}}##1}%
    \if@twocolumn
      \ifnum \col@number=\@ne
        \@maketitle
      \else
        \twocolumn[\@maketitle]%
      \fi
    \else
      \newpage
      \global\@topnum\z@   
      \@maketitle
    \fi
    \thispagestyle{plain}\@thanks
  \endgroup
  \setcounter{footnote}{0}%
}
\makeatother

\def\keywordname{{\bfseries \emph{Keywords}}}%
\def\keywords#1{\par\addvspace\medskipamount{\rightskip=0pt plus1cm
\def\and{\ifhmode\unskip\nobreak\fi\ $\cdot$
}\noindent\keywordname\enspace\ignorespaces#1\par}}
\date{today}

\begin{document}

\pagenumbering{roman}
\setcounter{page}{1}

\author[1]{\textbf{Owen Futter}}
\author[1]{\textbf{Nicola Mu\c{c}a Cirone}}
\author[2]{\textbf{Blanka Horvath}}

\affil[1]{\normalsize Imperial College London, Department of Mathematics}
\affil[2]{\normalsize University of Oxford, Mathematical Institute and Oxford Man Institute}

\title{\vspace{-0.8cm} \endgraf\rule{\textwidth}{.1pt} \\ \tiny \text{ } \\ \vspace{0.3cm} \textbf{\LARGE Kernel Learning for Mean-Variance Trading Strategies} \vspace{0.5cm} \\ \endgraf\rule{\textwidth}{.1pt}} 

\date{ }

\pagenumbering{arabic}
\setcounter{page}{1}

\maketitle

\vspace{-0.3cm}
\begin{abstract}
    In this article, we develop a kernel-based framework for constructing dynamic, path-dependent trading strategies under a mean-variance optimisation criterion. Building on the theoretical results of \cite{Cirone2025RoughHedging}, we parameterise trading strategies as functions in a reproducing kernel Hilbert space (RKHS), enabling a flexible and non-Markovian approach to optimal portfolio problems. 
    We compare this with the signature-based framework of \cite{Futter2023SignatureSignalsb} and demonstrate that both significantly outperform classical Markovian methods when the asset dynamics or predictive signals exhibit temporal dependencies for both synthetic and market-data examples.
    Using kernels in this context provides significant modelling flexibility, as the choice of feature embedding can range from randomised signatures to the final layers of neural network architectures. Crucially, our framework retains closed-form solutions and provides an alternative to gradient-based optimisation.
\end{abstract}
{
  \hypersetup{linkcolor=mediumtealblue}
  \tableofcontents
}

\section{Introduction} \label{sec:intro}

\subsection{Background and Motivation} \label{sec:motivation}

Trading strategy construction and portfolio choice is a fundamental problem in quantitative finance, where traders and researchers aim to balance maximising PnL with the associated constraints that they face such as the information they are in possession of, volatility, liquidity and the associated costs of trading. In this work, we are concerned with finding optimal \textit{dynamic}, \textit{path-dependent} trading strategies in which the past trajectory of information is incorporated into the decision making as new information filters in. In particular, we solve an optimal portfolio choice problem where the \textit{inventory (position)} is given as the control variate and derive a solution to the mean-variance criterion. 

Path-dependencies are ubiquitous in finance. They are present both in the underlying asset time series, either arising due to path-dependent volatility \cite{Guyon2014Path-DependentVolatility, Guyon2022VolatilityPath-Dependent} or due to price discovery (stochastic drift) as market participants react to diverse sources of information that arrives at differing speeds \cite{Toth2011AnomalousMarkets, Cont2014TheEvents, Ohara2015High}. Path-dependencies are also present in predictive signals, since they are often constructed using temporal interactions of input variables over time, meaning their alpha decay is not necessarily Markovian and dependent on how the information interacts with the market over time. We focus on an objective function which is dependent on the \textit{terminal time} PnL and variance under general conditions of the underlying asset process. The mean-variance criterion was originally posed by Markowitz \cite{Markowitz1952PortfolioSelection} and further studied in a vast body of literature, such as \cite{Li2000OptimalFormulation, Zhou2000Continuous-timeFramework, Lim2002Mean-VarianceMarket, Lim2004QuadraticMarket, Shen2015MeanvarianceCoefficients, AbiJaber2021MarkowitzModels, Han2021MeanVarianceModel} and the references therein. The inclusion of the terminal variance term forces the problem to become time-inconsistent when the underlying drift and volatility are stochastic processes, making the problem incompatible with classical dynamic programming techniques. A breakthrough was made when \cite{Lim2002Mean-VarianceMarket} extended the previous works of \cite{Li2000OptimalFormulation, Zhou2000Continuous-timeFramework} to account for stochastic drift and volatility in the underlying asset process, leading to an optimal trading strategy in terms of the solutions of a linear and non-linear BSDE. While this solution is in closed form, numerical schemes and BSDE solvers (such as parameterising the solution as a neural network) are required. Alternatively, one could assume a specific model for the underlying asset, for example in \cite{AbiJaber2021MarkowitzModels}, the volatility process is characterised as a Volterra process which is inherently path-dependent (making the problem non-Markovian), leading to explicit solutions. In this work, we propose an alternate, \textit{data-driven}, approach. We show that the presence of temporal structure in either the underlying asset or in predictive signals leads to Markovian solutions becoming sub-optimal and that by incorporating the memory of the market state, performance is improved.

In order to incorporate such path-dependencies, the control variable can be parameterised as a function of the past trajectory of information, which may include the asset price process, predictive signals and other market state variables such as order flow or measures of volatility. In this work, we extend the theoretical foundations of \cite{Cirone2025RoughHedging} in which a hedging strategy is parameterised as a function in a reproducing kernel Hilbert space (RKHS), providing a very general and flexible setting. Alternative formulations of these problems have been studied in \cite{Kalsi2020OptimalSignatures, Cartea2022Double-ExecutionSignatures, Futter2023SignatureSignalsb} where the control variable was parameterised as a linear functional on the signature. The authors in \cite{Akyildirim2023RandomizedSelection} also solve path-dependent portfolio optimisation problems in a data-driven way using \textit{Randomised Signatures} as a reservoir, while \cite{Cuchiero2023SignatureTheory} also use the signature to approximate solutions in \textit{Stochastic Portfolio Theory}.

The kernel framework that we apply in this work aims to be more general in the sense that a wide array of feature maps can be incorporated, including the signature transform. We also discuss how kernels provide more flexibility in dealing with a larger number of assets and exogenous variable inputs due to a greater computational efficiency arising from the \textit{kernel trick} 
, which is the ability to compute inner products in an infinite-dimensional feature space without explicitly computing the feature map.
Kernel methods have emerged as a powerful framework in machine learning and finance in order to model complex, nonlinear temporal relationships. The use of implicit feature maps allow to embed input trajectories into high-dimensional spaces without the need for explicit transformations. This approach underpins several widely used techniques, such as Support Vector Machines (SVMs) and Gaussian Processes (GPs) \cite{schölkopf2002learning, DeVito2004SomeMethods}.
Kernel methods are data-efficient and flexible and this work provides a powerful end-to-end framework for financial optimisation problems.

In general, we are concerned with optimal trading problems of the form 
\begin{align} \label{eq:1-general_objective}
    \argmax_{\phi \in \mathfrak{F} } \Ex^{\P} \left[ J\left( V_T^\phi \right) \right]
\end{align}
where the control variable $\phi$ is a progressively measurable, path-dependent function which parameterises the strategy position or trading speed, $V_T$ is the terminal PnL of the trading strategy, and $J(\cdot)$ is the chosen objective function.
It is then natural to ask, how must one find an optimal function $\phi$, and specifically how to choose a class of functions $\mathfrak{F}$ for which to search over. In this work, we utilise \textit{kernel methods} and the results of \cite{Cirone2025RoughHedging} to parameterise $\phi$ as a function in a RKHS, for which we will make concrete in Section~\ref{sec:2-kernel_trading}. That is to say, we let the control variable to be parameterised as
\begin{align} \label{eq:1-strategy_parameterisation}
    \phi(\psi(x)) = \langle \phi, K(\psi(x),\cdot)\rangle_{\H_K} \ \ \in \R^d
\end{align}
where $x$ represents the input stream of data and $\psi$ is a \textit{feature embedding}. $\psi$ maps the trajectories of $x$ into a new feature space, transforming our market state (filtration). The feature trajectories are then in-turn transformed into the trading strategy position or trading speed by the control function $\phi$. This formulation allows us to cast the optimisation in terms of $\phi$, in which we can optimise over all functions $\phi \in \cH_K$. This parameterisation is very general and flexible, since the choice of feature embedding $\psi$ can be as simple or as complex as one desires, depending on the problem at hand. We will discuss the form of the kernel $K$ and provide further examples in Section~\ref{sec:2-kernel_trading}.

In the specific case when the chosen kernel $K$ is the \textit{Signature Kernel} \cite{Kiraly2019KernelsData, Salvi2020ThePDE, Cass2021GeneralKernels}, our kernel based framework is in fact comparable to previous results presented in \cite{Futter2023SignatureSignalsb}.
A key advantage in this work is that signature kernels also enable us to utilise a \textit{kernel trick}, and consequently one need not compute the explicit signature (or even truncate it), but can instead solve a Goursat PDE, as outlined in \cite{Salvi2020ThePDE}.
We provide an extensive comparison in Section~\ref{sec:4-comparison_sig_trading} from both a strategy performance and a practical and computational perspective, with a summary presented in Section~\ref{subsec:5.3-comparison}. We find that there are distinct trade-offs computationally, as the signature computation and PDE solvers scale differently depending on (i) the dimension of the inputted feature trajectories, (ii) the trajectory length, and (iii) the number of trajectories used in the fitting procedure. We also note that from a performance perspective, the comparison is very much dependent on the problem at hand. The kernel framework of this paper aligns closely with
the broader machine learning literature, providing greater flexibility and generalisability. However, with this flexibility comes greater modelling complexity: although solutions are obtained in closed form, additional hyperparameters must be tuned (some of which must be computed prior to kernel evaluation), potentially leading to a bottleneck in fitting times\footnote{for further details see \cite{Alden2025SignatureTests}, as well as in Section~\ref{sec:5-implementation}.}. In Section~\ref{sec:5-implementation} we outline these practical considerations 
in order to ensure robustness and optimal performance. 
We note that in lower-dimensional settings or when higher order terms of the signature do not play a decisive role, the solution in \cite{Futter2023SignatureSignalsb} is likely to be preferable due to its tractability and explainability.
However, if the problem requires greater complexity, e.g. involving many input features/signals, assets, or non-linear dependencies, then the kernel method of this work is likely to be a more appropriate choice.

\subsection{Paper Outline} \label{sec:our_approach}

In this work, we aim to provide an accessible introduction into kernel methods, specifically in optimisation problems applied to finance, along with numerical examples that are able to demonstrate the importance of path-dependencies in such problems. Our contributions are as follows:
\begin{itemize}
\itemsep-0.5em
    \item In Section~\ref{sec:2-kernel_trading} we define the necessary tools for working with kernels and specifically define the key components of kernel trading strategies, as well as expressions for the PnL and variance of a trading strategy.
    \item In Section~\ref{sec:3-mean_variance} we derive a solution to the mean-variance objective when considering the strategy position as the control variable as well as a robust solution via spectral decomposition. We also discuss the relevance of path-dependence in our setting as well as a comparison to the classical Markowitz solution.
    \item In Section~\ref{sec:6-numerical_results} we provide numerical results comparing our framework against Markovian counterparts. We firstly consider a synthetic model in which the drift of the underlying asset is path-dependent. We then construct an experiment with market data and synthetic signals. Both setups yield a consistent outperformance of the compared Markovian models.
    \item In Section~\ref{sec:4-comparison_sig_trading} we provide an extensive comparison between the mean-variance solution in \cite{Futter2023SignatureSignalsb} in which the trading straetgy is parameterised as a linear functional on the signature. We compare the strategy performance and computational aspects for when the kernel used is the \textit{signature kernel}.
    \item In Section~\ref{sec:5-implementation} we discuss the practicalities and how to compute the optimal strategy, including hyperparameter optimisation.
\end{itemize}

\section{Kernel Trading Strategies} \label{sec:2-kernel_trading}

Firstly, we remind the reader of the most fundamental definitions and results for kernels, and provide further reading on kernels and rough path theory in Appendix~\ref{sec:appx_rough_paths}. For a thorough introduction to kernel methods see \cite{schölkopf2002learning}, and for examples in machine learning and finance, see \cite{pannier2024pathdependentpdesolverbased, Cirone2025RoughHedging}. Throughout this work we assume the following notation:
 \begin{itemize}
 \itemsep-0.3em
     \item $X=(X_t)_{t \in [0,T]}$ is a non-negative stochastic process satisfying $X_0^m=1$ for $m \in \{1,\dots,d\}$, which denotes the underlying \textit{tradable} asset process, which takes values in $V=\R^d$. We assume that the price paths $X_{0,t}$ are directly given as $\alpha$-Holder rough paths, where we assume that for any $t \in [0,T]$,  $X_{0,t} \in \mathcal{C}^\alpha([0,t],\R^d) =: \Omega_t^\alpha(\R^d)$.
     \item $\Lambda^\alpha_X := \bigcup_{t\in[0,T]} \Omega_t^\alpha(\R^{d})$ defines the space of trajectories of the underlying tradable asset process $X$, such that any trajectory satisfies $X_{0,t} \in \Lambda^\alpha$ for all $t \in[0,T]$. 
     \item $\psi : \Lambda^\alpha_X\to \Lambda^\psi := \bigcup_{t\in[0,T]} C^0([0,t],\R^M)$ is defined as a \textit{feature embedding} that maps the input asset trajectories into new streams of features, such that $\psi(X_{0,t})$ takes values in $\R^M$. 
     This feature embedding can be seen as a pre-processing modelling choice in order to balance generalisability and model complexity.
     \item $(\Omega_T^\alpha, \mathcal{B}(\Omega_T^\alpha), \mathbb{F}, \P)$ is our fixed filtered probability space, where $T > 0$ denotes the finite, deterministic terminal time, $\mathcal{B}(\Omega_T^\alpha)$ is the Borel $\sigma$-algebra and $\mathbb{F} = \{ \mathcal{F}_t : t \in [0,T] \}$ be the filtration generated by streams in $\Lambda^\psi$.
    \item \( \xi^m = (\xi^m_t)_{t \in[0,T]} \) denotes the traders \textit{position} (inventory) for asset $m \in \{1,\dots,d \}$. It is the control variable in the optimisation problems in Section~\ref{sec:3-mean_variance}.
 \end{itemize}

Building on the work of \cite{Cirone2025RoughHedging}, we focus on representing the control variable of a trading strategy as a function in a RKHS as seen in Equation \eqref{eq:1-strategy_parameterisation}. We refer to \textit{Kernel Trading Strategies}, as the collective of strategies that derive their values through such parameterisations, regardless of the kernel used.

\subsection{A Background of Kernels in Optimisation Problems} \label{subsec:2.1-kernel_optimisation}

Before we discuss the specifics of kernel trading strategies, we provide some very important results from the theory of kernels and how these results will allow us to solve the optimisation problems that we aim to solve, as well as providing intuition for the optimal solutions. Generally, the kind of RKHS that we are working with are respect to \textit{operator-valued kernels}, which are subtly distinct from \textit{scalar-valued kernels}. We refer to \textit{scalar-valued kernels} $\kappa$ of the form $\kappa : \Lambda^\psi \times \Lambda^\psi \to \R$, which maps a pair of trajectories to a scalar value, representative of a similarity score of the input trajectories. In our setting, we require to work with \textit{operator-valued kernels}, which allows us to produce solutions in a multi-dimensional vector-valued setting for when we have several assets to trade. 

\begin{definition}{(Operator Valued Kernel on $\Lambda^\psi$).}
    Given the set of feature trajectories $\Lambda^\psi$, an $\mathcal{L}(\R^d)$-valued kernel is a map $K: \Lambda^\psi \times  \Lambda^\psi \to \mathcal{L}(\R^d)$ such that
    \begin{enumerate}
    \itemsep-0.4em
        \item $K$ is symmetric: $\forall X,Y \in \Lambda^\psi$ it holds $K(X,Y) = K(X,Y)^T \in \R^{d \times d}$.
        \item $K$ is positive semi definite: For every $N \in \mathbb{N}$ and $\{(X_i,v_i)\}_{i=1,\dots,N}\subseteq \Lambda^\psi \times \R^d$ the matrix with entries $\sprod{K(X_i,X_j)v_i}{v_j}_{\R^d}$ is semi-positive definite.
    \end{enumerate}
    For such a kernel $K$, it is natural to associate a \textit{reproducing kernel Hilbert space} (RKHS) $\H_K$ which is a space of functions $\Lambda^\psi \to \R^d$, characterised by the following properties
    \begin{enumerate}[label=\roman*.]
    \itemsep-0.8em
        \item For every trajectory and vector pair $X, v \in \Lambda^\psi \times \R^d$, 
        \begin{align*}
            K(X,\cdot)v \in \H_K.
        \end{align*}
        \item For every trajectory and vector pair $X, v \in \Lambda^\psi \times \R^d$ and for all functions $\phi \in \H_K$, we have the \textit{reproducing property}
        \begin{align*}
            \langle \phi(\cdot), K(X, \cdot)v \rangle_{\H_K} = \langle \phi(X), v \rangle_{\R^d}.
        \end{align*}
        \item $\H_K$ is the closure, under $\langle \cdot, \cdot \rangle_{\H_K}$ of $\textup{Span}\{K(X,\cdot)v \ | \ X, v \in \Lambda^\psi \times \R^d\}$.
    \end{enumerate}
\end{definition}

\begin{remark}
    Intuitively, this gives us a concrete structure to the class of functions that we want to optimise over. In general, when applying a trading strategy, we can think of this as applying an “operator” on the assets that are being traded. In other words, we want to obtain $\langle \phi(\psi(X)), v \rangle_{\R^d}$ which is the inner product between the $d$ asset positions $\phi(\psi(X))$ and a vector $v \in \R^d$ which may correspond to the returns (increments) of the asset. The reproducing property of (ii) above allows us to express this as an inner product in the RKHS, which will allow us to solve the optimisation more cleanly. 
\end{remark}
\begin{example} \label{ex:2-operator-valued-kernel}
    We can in fact construct an \textit{operator-valued} kernel from an existing scalar-valued kernel with respect to some feature map $\varphi$. Let $\kappa_\varphi:\Lambda^\psi \times \Lambda^\psi \to \R$ be a scalar-valued kernel, and let $A$ be a symmetric positive definite matrix $A: \R^d \to \R^d$, then we can define an operator valued kernel as
    $$
        K(X,Y) = \kappa_\varphi(X,Y)A, \quad \forall X, Y \in \Lambda^\psi.
    $$
\end{example}
\begin{remark}
    We can define such a kernel $\kappa_\varphi$ through its feature map $\varphi : \Lambda^\psi \to E$, as $\kappa_\varphi(X,Y) = \langle \varphi(X), \varphi(Y)\rangle_E$. One particularly important choice of scalar-valued kernel $\kappa_\varphi$ is when $\varphi$ is given as the \textit{Signature} transform (Definition~\ref{defn:sig_transform}) denoted $\textup{Sig}(\cdot)$ and $E = T((\R^d))$. In this case, we would have
    \begin{align*}
        \kappasig(X,Y) = \langle \textup{Sig}(X), \textup{Sig}(Y) \rangle_{T((\R^d))}.
    \end{align*}
    In order to then convert this simply to an operator valued kernel, for our purpose in recovering a vector-valued control variable, we may have
    \begin{align*}
        K^{\textup{Sig}}_A(X,Y) = \kappasig(X,Y) \text{diag}(A).
    \end{align*}
    We may have for example that $A$ is simply the $d \times d$ identity matrix.
\end{remark}

As discussed, we aim to solve optimisations of the form in Equation~\eqref{eq:1-general_objective}, for which we wish to find an optimal function $\phi^* \in \H_K$. It is then natural to ask, what form should the optimal $\phi^*$ take? We can first state a classical result from kernel methods that will help us answer this question.

\begin{theorem}[General Functional Representer Theorem for Kernel Optimisation (Theorem 2, \cite{DeVito2004SomeMethods})]  \label{thm:2-general_optimisation}
    Let $\P$ be a probability measure on a locally compact second countable space $\X$. Let $J$ be a $p$-loss function (Definition \ref{defn:p-loss})
    with respect to $\P$, $p \in [1, \infty]$. Let $\mathcal{H}_K$ be a reproducing kernel Hilbert space such that the corresponding kernel $K$ is $p$-bounded with respect to $\P$. Define $q = p/(p-1)$.
    \[
    \phi^* \in \arg\min_{ \phi \in \mathcal{H}_K} \left\{ \ExP^X \left[ J(\phi(X)) \right] + \lambda \|\phi\|_{\mathcal{H}_K}^2 \right\}
    \]
    if and only if there is $\alpha^* \in L^q_\P(\X)$ satisfying
    \begin{align} \label{eq:2-optimal_alpha}
        \alpha^*(X) \in \partial U\big(\phi^*(X)\big) \quad X \in \X \text{ a.e.},
    \end{align}
    where the solution $\phi^*$ is given of the form
    \begin{align} \label{eq:2-optimal_phi_form}
    \phi^*(\cdot) = -\frac{1}{2\lambda}  \ExP^Y \left[ \alpha^*(Y) K(Y, \cdot) \right] \ \ \in \H_K    
    \end{align}
\end{theorem} 
\begin{proof}
    The proof is given thoroughly in \cite{DeVito2004SomeMethods}.
\end{proof}
\begin{remark}
    Here, $\alpha^* \colon \mathcal{X} \to \mathbb{R}$ is a measurable function (weight function) that can be determined by solving the optimisation problem \eqref{eq:2-optimal_alpha} with respect to the objective $U$. $K$ is the reproducing kernel that satisfies $K(X, \cdot) \in \mathcal{H}_K$ for all $X \in \mathcal{X}$. $\lambda \|\phi\|_{\mathcal{H}_K}^2$ is a regularisation term.
\end{remark}

\textbf{Notation}: We denote $\Ex^X_\P[\cdot]$ as the expectation over $X$ with respect to the probability measure $\P$, which may also be seen as $\Ex_{X\sim\P}[\cdot]$.

\begin{remark}
    Alternatively, in our setting we are instead \textit{maximising} some objective function and so the regularisation term effectively becomes negative as $-\lambda \|\phi\|_{\mathcal{H}_K}^2$, flipping the sign in \eqref{eq:2-optimal_phi_form}. This result is extremely powerful for us, since it allows us to understand functions $\phi \in \H_K$ in a more intuitive way and provides motivation to solve a more sophisticated class of objective functionals for trading. If for example, we want to evaluate $\phi$ with some “online” input trajectory $X$, we simply have the form of \eqref{eq:2-optimal_phi_form}, 
$$
\phi^*(X) = \frac{1}{2\lambda}  \ExP^Y \left[ \alpha^*(Y) K(Y, X) \right] \ \ \in \R^d
$$
which we are able to explicitly compute in practice given some reference trajectories $Y$ to take the expectation over. For further clarification, we give a specific example in Example~\ref{ex:appx_ridge_least_squares} of kernel ridge regression where the loss function is given as mean squared error. 
\end{remark}

\subsection{Inventory Formulation}

\begin{definition}{(Kernel Trading Strategy).} \label{def:2-kernel_strat}
    Let $X=(X_t)_{t \in [0,T]}$ be a $d$-dimensional tradable asset process taking trajectories in $\Lambda_X^\alpha$. Let $\psi : \Lambda_X^\alpha \to \Lambda^\psi$ be a feature embedding of the paths of $X$ (and other exogenous signals) for which $\psi(X)$ is a transformed trajectory taking values in $\R^M$. Let $K : \Lambda^\psi \times \Lambda^\psi \to \mathcal{L}(\R^d)$ be an operator-valued kernel. Let the position (inventory) process of a trading strategy be given as $\xi = (\xi)_{t \in [0,T]}$, which is an adapted, $\F_t$-predictable and integrable strategy such that $\int^T_0 \xi_s^2 ds < \infty$. We say $\xi$ is a \textit{kernel trading strategy} with respect to $K$ and $\psi$, if there exists a $\phi \in \H_K$ such that for any time $t \in [0,T]$,
    \begin{align*}
        \xi_t = \phi(\psi(X_{0,t})) = \sprod{\phi(\cdot)}{K(\psi(X_{0,t}), \cdot)}_{\H_K} \in \R^d,
    \end{align*}
    where $K(\psi(X_{0,t}), \cdot)$ is a \textit{canonical feature map} that allows to access non-linearities and path-dependencies in the trading strategy.
\end{definition}

We define the space of all \textit{kernel trading strategies} with respect to the kernel $K$ and feature embedding $\psi$ as
\[
\mathfrak{F}^{\psi}_K \coloneqq \left\{ \xi_t = \sprod{\phi(\cdot)}{K(\psi(X_{0,t}), \cdot)}_{\H_K} \forall t\in[0,T] \;\bigg|\; \int_0^T \xi_t^2 dt < \infty, \; \ \forall \text{ } X_{0,T} \in \Lambda^\alpha_X \text{ and } \phi \in \H_K \right\}.
\]

\begin{remark}
   While this definition may still seem abstract, due to the fact that we are dealing with inner products of functions and this may not seem immediately computable as such. However, this definition remains purposely general and could be made more specific by letting $\phi$ be of the form seen in \eqref{eq:2-optimal_phi_form}, which would yield
\begin{align*}
    \xi_t = \phi(\psi(X_{0,t})) = \Ex_{Y\sim\P} \left[\alpha(Y)  K\left( \psi(Y), \psi(X_{0,t}) \right) \right] \in \R^d,
\end{align*}
for some $\alpha: \Omega_T^\alpha \to \R$.
\end{remark}

\begin{remark}
    Firstly, the choice of \textit{feature embedding} $\psi : \Lambda^\alpha_X\to C^0([0,T] ; \R^M)$ is a modelling choice. Here, the embedded dimension $M$ could be greater or smaller than $d$, in order to either compress the market state space or lift it into a higher-dimensional representation. A simple example of $\psi$ would be to concatenate the asset process $X$ with exogenous predictive signals $\alpha$, e.g. $\psi(X) = (t,X,\alpha)$, enriching the filtration that the trader has access to.
    
    Secondly, the choice of kernel $K$ is also up to the user. 
    In this work we focus specifically on path-dependent kernels which are applied to time-series trajectories, such as the \textit{signature kernel} (Definition~\ref{defn:sig-kernel}). 
    Instead of computing the \textit{signature kernel} directly, one could also construct a kernel using \textit{randomised signatures} (Definition~\ref{defn:rand_signatures_comp}), for example 
    \begin{align*}
        k_{\text{r-Sig}}(x,y) = \langle \text{rSig}(x), \text{rSig}(y) \rangle_{\R^M}
    \end{align*}
    which can be seen as a more computationally cheap approximation of the signature kernel \cite{Cuchiero2021ExpressiveSignature, Cirone2023NeuralResNets}. This emphasises the flexibility of using kernels as there are many alternative ways of pre-processing and modelling choices that all fit into one very general framework and solution.
\end{remark}

When the control variable is given as the position (inventory), in order to compute the terminal PnL $V_T$ of the trading strategy (with respect to a given asset path $X$), we require an It\^o integral of the position $\xi$ against the underlying asset $X$. That is to say, we require a representation of the form 
\begin{align*}
    V_T(X) \equiv V_T = \int^T_0 \xi_t^\top dX_t = \int^T_0 \langle \xi_t , dX_t \rangle_{\R^d},
\end{align*}
in which to formulate the desired objective functionals.

\begin{theorem}[PnL of a Kernel Trading Strategy]
    Let $\xi = (\phi(\psi(X_{0,t})))_{t \in[0,T]} \in \mathfrak{F}^{\psi}_K$ be a kernel trading strategy. The corresponding PnL $V_T$ with respect to a given asset trajectory $X$ can be given as
    \begin{align} \label{eq:2-pnl_form}
        V_T(X) & = \sprod{\phi}{\Phi_X}_{\H_K} 
    \end{align}
    where 
    \begin{align} \label{eq:2-pnl_kernel_phi}
        \Phi_X(\cdot) = \int^T_0 K(\psi(X_{0,t}), \cdot) dX_t \ \in \H_K,
    \end{align}
    which represents a canonical “PnL” feature map $\Phi_X : \Lambda^\psi \to \R^d$ which sends input feature trajectories to PnLs, with respect to $X_{0,t}$. 
\end{theorem}
\begin{proof}
    This result follows from \cite[Section 2.1]{Cirone2025RoughHedging}, where we have that the equalities below are all well defined
    \begin{align}
        V_T & = \int^T_0 \xi_t^\top dX_t \nonumber \\
        & = \int^T_0 \sprod{\phi(\psi(X_{0,t}))}{dX_t} \nonumber \\
        & = \int^T_0 \sprod{\phi}{K(\psi(X_{0,t}), \cdot)dX_t }_{\H_K} \nonumber \\
        & = \sprod{\phi}{\int^T_0 K(\psi(X_{0,t}), \cdot) \ dX_t}_{\H_K} \nonumber \\
        & = \sprod{\phi}{\Phi_X}_{\H_K} \nonumber,
    \end{align}
    which follows from the reproducing property of the RKHS and the linearity of the inner product.
\end{proof}
\vspace{-0.2cm}
\begin{remark}
    We can see that $\Phi_X$ is in fact a new \textit{canonical feature map} representing a new “integral kernel” for paths. In general, such integrals can be represented as $\Phi : \Omega^\alpha_{T} \to \cH_{K}$ as 
$$\Phi(X_{0,T}) \equiv \Phi_X := \int_0^T K(\psi(X_{0,t}),\cdot)  dX_t \in \cH_{K}$$
which introduces a new kernel $\mathcal{K}_{\Phi}$ and associated RKHS $\cH_{\Phi}$,
\begin{equation}
    \begin{aligned}
         & \mathcal{K}_{\Phi} : \Omega^\alpha_T \times \Omega^\alpha_T \to \bR 
         \\
         & \mathcal{K}_{\Phi}(X_{0,T},Y_{0,T}) = \sprod{\Phi_X}{\Phi_{Y}}_{\cH_{K}}.
    \end{aligned}
\end{equation}
This is particularly useful, since we can now represent integrals of functions of kernels as functions of integrals of kernels, which is a subtle but powerful new tool in order to solve optimisation problems of this form. It is also worth noting that the inner product in \eqref{eq:2-pnl_form} remains in $\H_K$ even though it is with respect to a new feature map $\Phi_X$, this is simply since $\Phi_X$ is in $\H_K$ and $\H_{\Phi} \subset \H_K$. 
\end{remark}
\begin{example} \label{ex:2-phi_x_example}
To further understand this new kernel structure, we can evaluate, in the space $\H_\Phi$, the integral feature map $\Phi_X$ for the \textit{current} input trajectory of the asset $X_{0,T}$,
on a given “comparison” trajectory $Y_{0,T} \in \Omega_T^\alpha$.
This would be given as the quantity
\begin{align*}
\Phi_X(Y_{0,T}) &:= \mathrm{ev^{\Phi}_{Y_{0,T}}}(\Phi_X)
:= \langle \Phi_X, K_{\Phi}(Y_{0,T}, \cdot) \rangle_{\H_K} \\
& = \langle \Phi_X, \Phi_Y \rangle_{\H_K} \quad \left( = \mathcal{K}_{\Phi}(X_{0,T},Y_{0,T}) \right) \\
& = \int^T_{t=0} \sprod{ \underbrace{\int^T_{s=0}  K \left(\psi(X_{0,s}),\psi(Y_{0,t}) \right)dX_s }_{\in \R^d}} { \underbrace{dY_t}_{\in \R^d}} \ \in \R.
\end{align*}
This new kernel will allow to quantify each past trajectories impact on the current PnL, variance and other quantities of interest in our objective functions.

Note that the evaluation of $\Phi_{X}$ in $\H_K$ takes a different form: for any $v \in \R^d$ one has
\begin{align*}
    \sprod{\Phi_X(\psi(Y_{0,T}))}{v}_{\R^d} &:= \sprod{\mathrm{ev^{K}_{\psi(Y_{0,T})}}(\Phi_X)}{v}_{\R^d} 
    \\
    &:= 
    \langle \Phi_X, K(\psi(Y_{0,T}), \cdot)v \rangle_{\H_K}
    \\
    &=  \langle \int_0^T K(\psi(X_{0,t}),\cdot)dX_t, K(\psi(Y_{0,T}), \cdot)v \rangle_{\H_K}
    \\
    &=  \sprod{\int_{t=0}^T K(\psi(X_{0,t}),\psi(Y_{0,T})) dX_t}{v}_{\R^d}
\end{align*}
so that 
\begin{equation*}
    \Phi_X(\psi(Y_{0,T})) = \mathrm{ev}^{K}_{\psi(Y_{0,T})}(\Phi_X) = \int_{t=0}^T K(\psi(X_{0,t}),\psi(Y_{0,T})) dX_t \in \R^d.
\end{equation*}

In fact we can see that the two are related by
\begin{equation*}
    \Phi_X(Y_{0,T}) 
    = \mathrm{ev^{\Phi}_{Y_{0,T}}}(\Phi_X) 
    = 
    \int^T_{t=0} \sprod{ \mathrm{ev^{K}_{\psi(Y_{0,t})}}(\Phi_X) } { dY_t}_{\R^d}
    = 
    \int^T_{t=0} \sprod{ \Phi_X(\psi(Y_{0,t}))  } { dY_t}_{\R^d}. \\
\end{equation*}
\end{example}
\vspace{0.5cm}
\begin{lemma}[Expected PnL of a Kernel Trading Strategy]  \label{lem:2-exp_pnl}
    Let $\xi = (\phi(\psi(X_{0,t})))_{t \in[0,T]} \in \mathfrak{F}^{\psi}_K$ be a kernel trading strategy. The corresponding expected PnL with respect to a probability measure $\P$ can be given as
    \begin{align}
        \ExP^X \left[V_T^\phi (X) \right] = \sprod{\phi}{\ExP^X[\Phi_X]}_{\H_K}.
    \end{align}
    In practice, the probability measure $\P$ will be the empirical measure of the observed data.
\end{lemma}

\begin{lemma}[Variance of PnL for a Kernel Trading Strategy] \label{lem:2-var_pnl}
    Let $\xi = (\phi(\psi(X_{0,t})))_{t \in[0,T]} \in \mathfrak{F}^{\psi}_K$ be a kernel trading strategy. The variance of the PnL $V_T$ is given as
    \begin{align*}
        \textup{Var}_{X\sim \P} \left[ V_T^\phi(X) \right] = \sprod{\phi}{\ExP^X \left[\Phi_X \otimes \Phi_X \right] \phi }_{\H_K} - \sprod{\phi}{\ExP^X[\Phi_X]}_{\H_K}^2.
    \end{align*}
    This result follows simply from the fact that the variance can be given as $\textup{Var}^{\P}(V_T^\phi) = \Ex^{\P}[(V_T^\phi)^2] - (\Ex^{\P}[V_T^\phi])^2$.
    Here we used the tensor product to write $\ExP^X \left[\Phi_X \otimes \Phi_X \right] \phi := \ExP^X \left[\Phi_X \sprod{\Phi_X}{\phi}_{\H_K} \right]$.
\end{lemma}
\section{Path-Dependent Mean-Variance Optimisation with Kernels
} \label{sec:3-mean_variance}
In this section we derive our main result for when the control variable is given as the inventory $\xi$ of the trading strategy. In kernel methods (and machine learning in general) it is customary to include a regularisation term on $\phi$ within the objective function due to the function space being so large and complex combined with access to finite data, this ensures that the chosen function exists and is not overfit. 
Hence, going forward we will generally include a regularisation term of the form $\norm{\phi}_{\cH_{K}}$. We explicitly derive a solution to the classical mean-variance criterion, which is given as
\begin{align} \label{eq:3-mean_var_objective}
    \max_{\xi \in \mathfrak{F}^\psi_K} \left\{ \ExP^X \left[V_T^\xi(X) \right] - \frac{\eta}{2} \textup{Var}_{\P}^X\left[ V_T^\xi(X) \right] \right\}
\end{align}
where $\eta$ is a risk-aversion parameter. Equivalently, we can formulate this as
\begin{align*}
    \max_{\xi \in \mathfrak{F}^\psi_K} \left\{ \ExP^X \left[ V_T^\xi(X) - \frac{\eta}{2} \left(  V_T^\xi(X)^2 - \ExP^X \left[ V_T^\xi(X) \right]^2\right) \right]  \right\}.
\end{align*}
Now, given that we also know from Theorem~\ref{eq:2-pnl_form} that we can express the PnL as $V_T^\xi(X) = \sprod{\phi}{\Phi_X}_{\H_K} \, = \mathrm{ev}^{\Phi}_{X_{0,T}} (\phi)$, then we can, adding regularization, in fact represent this optimisation over functions $\phi \in \H_K$ as
\begin{align*}
    \max_{\phi \in \H_K}  \left\{ \ExP^X \left[ \sprod{\phi}{\Phi_X}_{\H_K} - \frac{\eta}{2} \left( \sprod{\phi}{\Phi_X}_{\H_K}^2 - \ExP^X \left[ \sprod{\phi}{\Phi_X}_{\H_K} \right]^2 \right) \right] - \lambda \|\phi\|_{\mathcal{H}_K}^2 \right\}
\end{align*}
which further reduces to 
\begin{align} \label{eq:3-RKHS_objective}
    \max_{\phi \in \H_\Phi}  \left\{ \ExP^X \left[ \sprod{\phi}{\Phi_X}_{\H_\Phi} - \frac{\eta}{2} \left( \sprod{\phi}{\Phi_X}_{\H_\Phi}^2 - \ExP^X \left[ \sprod{\phi}{\Phi_X}_{\H_\Phi} \right]^2 \right) \right] - \lambda \|\phi\|_{\mathcal{H}_\Phi}^2 \right\}.
\end{align}
There is a subtle difference here, in the sense that we shift the focus from the space $\H_K$ to $\H_\Phi$ 
as $\H_K = \H_\Phi \oplus \H_\Phi^{\top}$ and $\phi \in \H_K$ enters in the equation only via its scalar products with elements of $\H_\Phi$.
$\H_\Phi$.
The natural question to ask now is, what form should the optimal strategy $\phi^*$ take? 

\subsection{Solution}

\begin{restatable}{theorem}{mvthm}\textup{(Optimal Kernel Trading Strategy for the Mean-Variance Criterion).}  \label{thm:3-mean_var_solution}
    Let $\xi \in \mathfrak{F}^{\psi}_K$ be a kernel trading strategy with respect to a kernel $K$ and feature embedding $\psi$. Let 
    $\X \subset \Omega_T^{\alpha}$
    be a locally compact subset trajectories, such that the kernel
    $\K_\Phi : \X \times \X \to \R$
    is measurable 
    with respect to $\mathbb{P}$.
    Assume moreover the second moment condition
    \begin{equation*}
        \ExP^{X,Y}\left[ \norm{\Phi_X}^2 \right] < +\infty
    \end{equation*}
    and the map $\Omega: L^2_\P(\X) \to L^2_\P(\X)$
    given by 
    \begin{equation*}
        \alpha \mapsto \ExP^X\left[ \alpha(X) \K_\Phi(X, \cdot)\right]
    \end{equation*}
    to be \emph{coercive} and \emph{invertible}.
    Then, we have 
    \begin{align*}
        \xi^* \in \argmax_{\xi \in \mathfrak{F}^\psi_K} \left\{ \Ex^\P \left[V_T^\xi \right] - \frac{\eta}{2} \textup{Var}^\P(V_T^\xi)  - \frac{\lambda}{2} \|\xi\|^2 \right\}.
    \end{align*}
    if and only if $\xi^*_t = \phi^*(\psi(X_{0,t})) 
    \, = \mathrm{ev}^K_{\psi(X_{0,t})}
    (\phi^*) \in \R^d
    $
    where
    \begin{align*}
        \phi^* = \ExP^X \left[ \alpha^*(X) \Phi_{X}\right] \in \H_K, \quad \norm{\xi}^2 = \langle \phi, \phi \rangle_{\H_K}
    \end{align*}
    and there exists an $\alpha^* \in L^2_\P(\X)$ that satisfies
    \begin{align*}
        \left(\lambda \textup{Id} + \eta \Xi^\Phi_{\P} \right)(\alpha^*)(X) = 1
    \end{align*}
    where $\Xi^\Phi_\P : L_\P^2 \to L_\P^2$ represents a covariance operator, defined as
    \begin{align*}
        \Xi^\Phi_{\P}(\alpha)(X) = \ExP^Y[\alpha(Y) \K_{\Phi}(Y,X)] - \ExP^{Y, Z}[\alpha(Y) \K_{\Phi}(Y, Z) ]\
    \end{align*}
\end{restatable}
\begin{proof}
    Given in Appendix~\ref{sec:appx_mean_var_proof}. 
\end{proof}

\begin{remark}
    We can see how similar this result is to that of \cite[Theorem 1]{Cirone2025RoughHedging}, where the form of $\phi^*$ is the same and the use of the new PnL feature map $\Phi_X$ is incorporated, however the solution for $\alpha^*$ is obviously different. 
\end{remark}

Now, we know that our original definition of a kernel trading strategy was defined on $K$ and not $K_\Phi$, which was instead used in the fitting procedure above. Therefore, we must explicitly compute the positions of the trading strategy $\xi_t$ at any time $t$, given an input feature trajectory $\psi(X_{0,t})$. 

\begin{corollary}[Optimal Kernel Trading Strategy Position]
    Suppose $\xi^* \in  \mathfrak{F}^{\psi}_K$ is the optimal kernel trading strategy corresponding to $\phi^* \in \H_K$ as computed in Theorem~\ref{thm:3-mean_var_solution}. Then, given a current input feature trajectory $\psi(X_{0,t})$ at time $t$, the position for the $m$-th asset $\xi_t^m$ is given as
    \begin{equation}
        \begin{aligned}
            \xi_t^m & = [\phi^*(\psi(X_{0,t}))]_m \\
             &= \sprod{\phi^*(\psi(X_{0,t}))}{e_m}_{\R^d}
            \\
            &
            =
            \sprod{\ExP^Y \left[ \alpha^*(Y) \Phi_{Y}(\psi(X_{0,t}))\right]}{e_m}_{\R^d}
            \\
            & = \frac{1}{\lambda} \ExP^Y \left[ \alpha^*(Y) \sprod{\Phi_{Y}}{K(\psi(X_{0,t}),\cdot)e_m}_{\cH_{K}} \right]
            \\
            & = \frac{1}{\lambda} \ExP^Y \left[ \alpha^*(Y) 
            \int_{t=0}^T \sprod{K(\psi(Y_{0,s}),\cdot) dY_s}{K(\psi(X_{0,t}),\cdot)e_m}_{\cH_{K}} \right]
            \\
            & = \frac{1}{\lambda} \ExP^Y \Bigg[\alpha^*(Y) 
            \underbrace{\int_{t=0}^T \sprod{K(\psi(X_{0,t}), \psi(Y_{0,s}))e_m}{ dY_s }_{\R^d}}_{\Gamma(X_{0,t},Y) \in \R} \Bigg] \\
            & = \frac{1}{\lambda} \ExP^Y \left[ \alpha^*(Y) \Gamma(X_{0,t},Y) \right]
        \end{aligned}
    \end{equation}
    for all $m \in \{1,\dots, d\}.$
\end{corollary}

\begin{remark}
    When, we are working with the empirical measure (as is the case in practice), we have access to $N$ \textit{co-location} trajectories of $X$, which we denote $\bfX = \{X_1, \dots, X_N\}$ then $L^2_{\P}(\cX)$ is identified with $\bR^n$ via $\delta_{X_i} \mapsto e_i$.
    Then the covariance operator $\Xi^{\Phi}_{\P}$ can be identified through a combination of the Gram matrix of $\K_{\Phi}$, i.e
    \begin{align*}
        \Xi^{\Phi}_{\P}(\cdot)(\bfX) = \frac{1}{N} \K_{\Phi}(\bfX, \bfX) \left(\textup{Id}_N - \frac{1}{N} \boldsymbol{1}_N \boldsymbol{1}_N^\top \right) \ \ \in \R^{N\times N}.
    \end{align*}
    This then allows us to explicitly obtain $\alpha^*(\bfX) \in \R^N$, as
    \begin{align} \label{eq:3-mean-var_alpha_star}
        \alpha^*(\bfX) = \left( \lambda \text{Id}_N + \frac{\eta}{N} \K_{\Phi}(\bfX, \bfX)\left(\textup{Id}_N - \frac{1}{N} \boldsymbol{1}_N \boldsymbol{1}_N^\top \right) \right)^{\dag}  \boldsymbol{1}_N.
    \end{align}
     We can also define a map $\Gamma_{\P} : \Lambda^\psi \to \bR^{d \times N}$ which intuitively computes the expected future PnL for each asset by comparing the current trajectory to all \textit{co-location trajectories}. We define this map explicitly for the $m$-th asset and $i$-th sample trajectory
    \begin{align} \label{eq:3-gamma_form}
    [ \Gamma_{\P}(\psi(X_{0,t}))]_{m,i} := \int_{t=0}^T \sprod{K(\psi(X_{0,t}), \psi(Y_{0,s}^i))e_m}{dY_s^i}_{\bR^d} \ \ \in \R. 
    \end{align}
    We denote the mapping over all sample trajectories and assets as $\Gamma_\P(\psi(X_{0,t})) \in \R^{d\times N}$. Then, as above we can explicitly compute the trading strategy positions as
    \begin{equation}\label{eq:3-strat_xi_form}
    \xi_t = \phi^*(\psi(X_{0,t})) =   \underbrace{ \frac{1}{N \lambda} \Gamma_{\P}\left(\psi(X_{0,t}) \right)}_{\in \R^{d \times N}} \underbrace {\left( \lambda \text{Id}_N + \frac{\eta}{N} \K_{\Phi}(\bfX, \bfX)(\textup{Id}_N - \frac{1}{N} \boldsymbol{1}_N \boldsymbol{1}_N^\top) \right)^{\dag}  \boldsymbol{1}_N}_{\text{Optimal Weights } \alpha^* \in \R^N}
    \end{equation}
\end{remark}
where $\dag$ represents the Moore–Penrose pseudo-inverse. Here, the positions $\xi_t$ are made up of two components:
\begin{itemize}
\vspace{-0.2cm}
\itemsep-0.5em
    \item Linear weights $\alpha^* \in \R^N$ which are fitted in a training phase using a transformed version of the gram matrix of the PnL feature map $\Phi_X$.
    \item The function $\Gamma_{\mathbb{P}}(\psi(X_{0,t}))$, which dynamically evaluates the expected PnL 
    if we had traded based on the output of the kernel with all other co-located trajectories. 
\end{itemize}
\begin{figure}
    \centering
    \includegraphics[width=\linewidth]{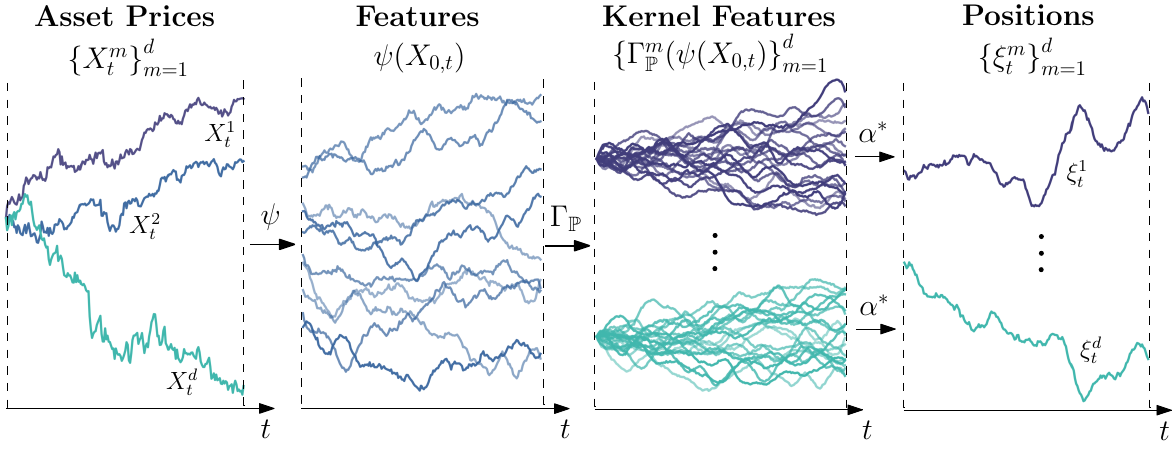}
    \caption{\centering A visual depiction of how data is passed through the Kernel Trading framework.}
    \label{fig:online_phase}
\end{figure}

\subsection{Expected PnL \& Variance}

\begin{corollary}[Expected PnL of a Kernel Trading Strategy]
    Let $\xi \in \mathfrak{F}^{\psi}_K$ be a kernel trading strategy such that 
    $\xi_t = \phi(\psi(X_{0,t}))$
    where $\phi$ is of the form
    \begin{align*}
        \phi^* = \ExP^Y \left[ \alpha(Y) \Phi_{Y} \right]
    \end{align*}
    for some $\alpha \in 
    L^2_\P(\X)
    $. Then the terminal PnL $V_T$ of this strategy over a given input trajectory $X_{0,T}$ is given as
    \begin{align*}
        V_T^\alpha(X) & = \langle \phi, \Phi_X \rangle_{\H_K} \\
        & = \ExP^Y\left[ \alpha(Y) K_\Phi(Y, X) \right]
        \ \ \in \R.
    \end{align*}
    Now, if we have access to $N$ \textit{co-location} trajectories of  $Y$ in which we can take this expectation, denoted $\bfY = \{Y^1, \dots, Y^N\}$, then we have the empirical expectation
    \begin{align*}
        V_T^\alpha(X) & = \frac{1}{N} \sum^N_{i=1} \alpha(Y_i) \K_\Phi(Y^i, X) \\
        & = \frac{1}{N} \underbrace{\alpha^\top}_{\in \R^N} \underbrace{\K_\Phi(\bfY, X)}_{\in \R^N} \ \ \in \R
    \end{align*}
    We want to take the expectation of the PnL over a new set of $M$ sample trajectories of $X$, denoted $\bfX = \{X^1, \dots, X^N\}$. That is, we have
    \begin{align*}
        \Ex_{X\sim\bfX}[V_T^\alpha(X)] & = \Ex_{X\sim\bfX}\left[ \Ex_{Y \sim \bfY} \left[ \alpha(Y) K_\Phi(Y, X) \right] \right] \\
        & = \frac{1}{NM} \underbrace{\alpha^\top}_{\in \R^N} \underbrace{\K_{\Phi}(\bfY, \bfX)}_{\in \R^{N \times M}} \mathbf{1}_M \ \ \in \R.
    \end{align*}
    We are able to imagine that perhaps $\bfY$ is a \textit{sample} set of trajectories and $\bfX$ is the out of sample \textit{test} set of trajectories. In this case, the expectation over $\bfX$ corresponds to the expected PnL out of sample, if instead we wanted the in sample expected PnL, we can simply let $\bfX = \bfY$.
\end{corollary}
\begin{corollary}[Variance of Kernel Trading Strategy PnL]
    Let $\xi \in \mathfrak{F}^{\psi}_K$ be a kernel trading strategy such that 
    where $\phi$ is of the form
    \begin{align*}
        \phi^* = \ExP^Y \left[ \alpha(Y) \Phi_{Y}\right]
    \end{align*}
    for some $\alpha \in L^2_\P(\X)$. We can similarly define the variance over a set of trajectories $\bfX = \{X^1, \dots, X^N\}$ as
    \begin{align*}
        \textup{Var}_{X \sim \bfX}(V_T^\alpha(X)) & = \Ex_{X \sim \bfX}[(V_T^\alpha(X))^2] - (\Ex_{X \sim \bfX}[V_T^\alpha(X)])^2 \\
        & = \Ex_{X \sim \bfX} \left[ \Ex_{Y \sim \bfY} \left[ \alpha(Y) K_\Phi(Y, X) \right]^2 \right] - \left( \Ex_{X\sim\bfX}\left[ \Ex_{Y \sim \bfY} \left[ \alpha(Y) K_\Phi(Y, X) \right] \right] \right)^2 \\
        & = \frac{1}{M}\frac{1}{N^2} \alpha^\top \K_\Phi(\bfY, \bfX) \K_\Phi(\bfY, \bfX)^\top \alpha - \frac{1}{(NM)^2} \alpha^\top \K_{\Phi}(\bfY, \bfX) \mathbf{1}_M \mathbf{1}_M^\top \K_{\Phi}(\bfY, \bfX)^\top \alpha \\
        & = \frac{1}{N^2 M} \alpha^\top \K_{\Phi}(\bfY, \bfX) \left( \textup{Id}_M - \frac{1}{M}\mathbf{1}_M \mathbf{1}_M^\top \right) \K_{\Phi}(\bfY, \bfX)^\top \alpha.
    \end{align*}
\end{corollary}
\begin{remark}
    We can in fact see that there is a direct link between the optimal weights $\alpha^*$ computed in Theorem~\ref{thm:3-mean_var_solution} and the mean/variance in the previous corollaries. We previously computed the optimal $\alpha^*$ in \eqref{eq:3-mean-var_alpha_star}.
    Setting $\bfX = \bfY$ in the previous corrollaries, we obtain the sample expected PnL and variance as
    \begin{align*}
        \Ex_{X\sim\bfX}[V_T^\alpha(X)] & = \frac{1}{N^2} \alpha^\top \K_{\Phi}(\bfX, \bfX) \mathbf{1}_M \\
        & = \alpha^\top \mu_{\Phi} \\
        \textup{Var}_{X \sim \bfX}(V_T^\alpha(X)) & = \frac{1}{N^3} \alpha^\top \K_{\Phi}(\bfX, \bfX) \left( \textup{Id}_M - \frac{1}{M}\mathbf{1}_M \mathbf{1}_M^\top \right) \K_{\Phi}(\bfX, \bfX) \alpha \\
        & = \alpha^\top \Sigma_\Phi \alpha.
    \end{align*}
    We can then compute the form analagous to classical mean-variance,
    \begin{align*}
        (\Sigma_\Phi)^{-1}\mu_{\Phi} & = N \left( \K_{\Phi}(\bfX, \bfX) \left( \textup{Id}_M - \frac{1}{M}\mathbf{1}_M \mathbf{1}_M^\top \right) \K_{\Phi}(\bfX, \bfX) \right)^{-1} \K_{\Phi}(\bfX, \bfX) \mathbf{1}_M \\
        & = N \left( \K_{\Phi}(\bfX, \bfX) \left( \textup{Id}_M - \frac{1}{M}\mathbf{1}_M \mathbf{1}_M^\top \right) \right)^{-1} \mathbf{1}_M,
    \end{align*}
    where we obtained cancellations of the Gram matrix, since if two matrices $A$ and $B$ are symmetric and invertible, then $(ABA)^{-1}A = (AB)^{-1}$. We can clearly see that this term is very similar to the optimal weights computed in \eqref{eq:3-mean-var_alpha_star}, hence our solution retains the classical mean-variance structure of the inverse of a covariance multiplied by the expected return. 
\end{remark}

\subsection{A More Robust $\alpha$} \label{subsec:3-robust_alpha}
We previously found the solution $\alpha^*$ in \eqref{eq:3-mean-var_alpha_star} as
\begin{align} \label{eq:3-alpha_v2}
    \alpha^*(\bfX) = \left( \lambda \text{Id}_N + \frac{\eta}{N} \K_{\Phi}(\bfX, \bfX)\left(\textup{Id}_N - \frac{1}{N} \boldsymbol{1}_N \boldsymbol{1}_N^\top \right) \right)^{\dag}  \boldsymbol{1}_N.
\end{align}
However, it is common in kernel methods that the matrix that is being inverted is close to singular and not always well-conditioned, causing numerical instability even when using a pseudo-inverse. Alternatively, we are able avoid the inversion altogether, and use only dominant eignenvectors, explicitly controlled by the eigenvalues $\gamma_k$, so small values can be ignored. 
\begin{restatable}{theorem}{alphathm}\textup{(Spectral Representation of $\alpha^*$).}  \label{thm:3-spectral_alpha}
Let $\K_\Phi(\bfX, \bfX) \in \mathbb{R}^{N \times N}$ be a symmetric positive semidefinite gram matrix, and define
\begin{align*}
A := \lambda \textup{Id}_N + \frac{\eta}{N} \K_{\Phi}(\bfX, \bfX).
\end{align*}
 We define the unit-norm vector as $\mathbf{e}_N := \sqrt{N}^{-1} (1,1,\dots,1)^{\top}$, such that we have $\frac{1}{N} \boldsymbol{1}_N \boldsymbol{1}_N^\top =\mathbf{e}_N \mathbf{e}_N^\top$. Then the solution to the system (from Theorem~\ref{thm:3-mean_var_solution})
\begin{align*}
\left( \lambda \textup{Id}_N + \eta \Xi^\Phi_{\P} \right)(\alpha^*) = \boldsymbol{1}_N,
\end{align*}
where
\begin{align*}
\Xi^{\Phi}_{\P}(\cdot)(\bfX) = \frac{1}{N} \K_{\Phi}(\bfX, \bfX) \left(\textup{Id}_N - \frac{1}{N} \boldsymbol{1}_N \boldsymbol{1}_N^\top \right),
\end{align*}
is given by
\begin{align*}
\alpha^* = \frac{\sqrt{N} A^{-1} \mathbf{e}_N}{\lambda \mathbf{e}_N^{\top} A^{-1} \mathbf{e}_N},
\end{align*}
and is equivalent to that of \eqref{eq:3-alpha_v2}. Moreover, if \( A \) admits an orthogonal eigen-decomposition \( A = \sum_{k=1}^N \gamma_k \mathbf{u}_k \mathbf{u}_k^{\top} \), then
\begin{align} \label{eq:3-alpha_svd}
\alpha^* = \frac{\sqrt{N}}{\lambda} \left( \sum_{k=1}^N \frac{(\mathbf{u}_k^\top \mathbf{e}_N)^2}{\gamma_k} \right)^{-1} \sum_{k=1}^N \frac{\mathbf{u}_k^\top \mathbf{e}_N}{\gamma_k} \mathbf{u}_k.
\end{align}
\end{restatable}
\begin{proof}
    Given in Appendix~\ref{sec:spectral_alpha_proof}.
\end{proof}
\begin{remark}
In practice, we may truncate the eigenvalue decomposition to the top $m<N$ terms to obtain a low-rank approximation of $\alpha^*$, which improves computational efficiency and numerical stability. This is especially useful when \(A\) is nearly low-rank or $N$ is large. It may also be useful to select an eigenvalue threshold, and only select eigenvectors corresponding to the eigenvalues that are greater than such a threshold.
\end{remark}
We find that this robust solution is often \textit{essential} in order to maximise the potential of the kernel methods. Often, even when accounting for appropriate scaling of paths (see further in Figure~\ref{fig:scaling}), the inversion becomes too unstable to be useful for small values of $\lambda$. We provide further evidence in Section~\ref{sec:5-implementation} and in Figure~\ref{fig:svd_noisiness}.

\subsection{Instantaneous Variance Constraint \& Markowitz Comparison} \label{subsec:stoch_drift_markowitz}

Previously, we found a solution to the objective in \eqref{eq:3-mean_var_objective} which penalises the \textit{terminal} variance of the PnL, which is used ubiquitously in mean-variance problems such as \cite{Basak2011DynamicSolution,Lefebvre2020Mean-variancePenalization, Jaber2021MarkowitzModels, Zhu2021OptimalObjective} as well as in the signature framework of \cite{Futter2023SignatureSignalsb}. Alternatively, instead of using the terminal variance constraint, the likes of \cite{Garleanu2013DynamicCosts, Jaber2024OptimalPropagators} have adopted a local instantaneous variance constraint, that ignores time-dependent variance either from the underlying drift/signal or from the underlying volatility of the asset. We aim to explore and make concrete how these two approaches yield different solutions, and how the choice of the objective is dependent on the traders constraints and the problem at hand. The local time variance constraint replaces the $\text{Var}(V_T)$ term in Equation \eqref{eq:3-mean_var_objective}, with a penalty
\begin{align*}
    \V^\xi_T(X) = \int^T_0 \xi_t^\top \Sigma \xi_tdt
\end{align*}
where $\Sigma \in \R^{d \times d}$ is a symmetric nonnegative definite covariance matrix. We make its motivation precise in the following example. 

\subsubsection*{Example: Stochastic Drift \& Static Volatility}

Suppose there is exploitable temporal structure in the underlying asset through a stochastic drift term $\mu_t$, that a trader wishes to profit from. Consider the $d$-dimensional asset process $X = (X^1,\dots,X^d)$, where
\begin{align*}
    dX_t = \mu_t dt + dM_t
\end{align*}
where $M$ is a continuous martingale such that
\begin{align*}
d[M^m, M^n]_t = \Sigma_{mn} dt, \quad \text{for } m, n = 1, \dots, d,
\end{align*}
and \( \Sigma\in \mathbb{R}^{d \times d} \) is a symmetric nonnegative definite covariance matrix. The mean-variance objective with local instantaneous variance is given as
\begin{align} \label{eq:3-markowitz_obj}
    J_{\text{M}}(\xi) = \Ex \left[ \int_0^T \xi_t^\top \mu_t dt \right] - \frac{\eta}{2} \Ex \left[ \int_0^T \xi_t^\top \Sigma \xi_t  dt \right] 
\end{align}
which we denote $J_{\text{M}}(\cdot)$, commonly referred to as the \textit{Markowitz} objective. The optimal strategy of the Markowitz objective in \eqref{eq:3-markowitz_obj} is famously given as
\begin{align} \label{eq:3-markowitz_optimal}
    \xi_t^{\text{M}} & = \frac{1}{\eta} \Sigma^{-1} \Ex_t [\mu_t], \ \ \ \forall t \in[0,T].
\end{align}
Meanwhile, we can reconstruct the objective in Equation \eqref{eq:3-mean_var_objective} as,
\begin{align} \label{eq:3-our_objective}
     \notag J(\xi) & = \Ex \left[ \int^T_0 \xi_t dX_t \right] - \frac{\eta}{2} \text{Var} \left[ \int^T_0 \xi_t dX_t \right] \\
    \notag & = \Ex \left[ \int^T_0 \xi_t^\top \mu_t dt\right] - \frac{\eta}{2} \text{Var}\left[ \int^T_0 \xi_t^\top \mu_t dt + \int^T_0 \xi_t^\top dM_t  \right] \\
     & =  \underbrace{\Ex \left[ \int^T_0 \xi_t^\top \mu_t dt\right] - \frac{\eta}{2} \Ex \left[ \int_0^T \xi_t^\top \Sigma \xi_t  dt \right]}_{J_{\text{M}}(\xi)} \\
     \notag & \quad \quad - \underbrace{\frac{\eta}{2} \text{Var} \left[ \int^T_0 \xi_t^\top \mu_t dt \right] - \eta \text{Cov}\left( \int^T_0 \xi_t^\top \mu_t dt,   \int^T_0 \xi_t^\top dM_t \right)}_{\text{Remainder}}, \\
\end{align}
where the remainder term corresponds to a hedging demand induced by intertemporal changes in expected return, as seen in \cite{Kim1996DynamicBehavior}.
 Both of the objectives in \eqref{eq:3-markowitz_obj} and \eqref{eq:3-our_objective} are similar, but differ by the remainder term.

In order to maximise this objective, after 
taking
the functional derivative of $J(\xi)$, we obtain a system of integral equations which are not able to be solved explicitly due to the dependence on $M_t$. The problem becomes much simpler if $\mu$ and $M$ were independent (e.g. as an inhomogeneous OU process), then the optimal strategy $\xi_t$ would satisfy the system 
\begin{align*}
    \Ex_t[\mu_t] - \eta \Sigma \xi_t -  \Ex_t[\mu_t] \int_0^T \left( \Ex_t[\xi_s^\top \mu_s] - \Ex[\xi_s^\top \mu_s] \right) ds = 0 \quad \forall t\in[0,T].
\end{align*}

This means that the optimal solution of Markowitz in \eqref{eq:3-markowitz_optimal} becomes \textit{sub-optimal} for the terminal time variance objective of \eqref{eq:3-mean_var_objective}, and the optimal solution is instead dependent on the expectation of the future drift terms $\mu_s$ for $s\in[0,T]$. This kind of result, in which $\xi_t^*$ is dependent on the expectation of the future signals decay, is in fact similar to the solutions obtained when transient market impact are included in the objective, as seen in \cite{Lehalle2019IncorporatingTrading, Jaber2025OptimalCase, Hey2023TradingEvidence, Webster2024HandbookModeling, Jaber2024OptimalPropagators}. This will provide further motivation for the outperformance that is obtained in Sections~\ref{sec:6-numerical_results} and \ref{sec:4-comparison_sig_trading}. We provide further details in Appendix~\ref{sec:appx-solution-variance}. For a solution to this problem with general stochastic drift and volatility, \cite{Lim2002Mean-VarianceMarket} obtain a strategy in terms of the solution of two BSDEs.
\begin{remark}
    It is important to note that in a pure (gross) “Sharpe ratio” sense, one cannot beat the optimal Markowitz strategy in this model above, since there is no stochastic volatility and that the variance of the increments of $X$ will be static through time. For low-frequency trading strategies, terminal variance is typically less of a priority. Instead, minimising daily volatility (and, by extension, drawdowns) is often more relevant. This suggests that a static variance constraint may be better suited to these types of strategies. On the other hand, the terminal variance penalty approach may be more appropriate for intraday traders 
    since volatility is stochastic, exhibiting intraday seasonality, which is naturally captured within our path-dependent framework. Therefore, there appears to be a trade-off: one may tolerate greater instantaneous volatility along the path in order to exploit path-dependencies that reduce terminal variance. This is conceptually similar to the trade-off introduced when optimising under transaction costs, where smoother trading increases the variance of the position process, but reduces costs, which may be the more critical constraint. Naturally, a trader will tailor their objective to include the constraints that they face in reality, one of which may be the variance of the terminal PnL.
\end{remark}
We note that the inclusion of the local time variance constraint is possible in the kernel trading framework, such that we obtain
\begin{align} \label{eq:3-kernel_inst_variance}
    \V_T^\xi(X) = \sprod{\phi}{\Omega_\P \phi}_{\H_K}
\end{align}
where $\Omega_\P : \H_K \to \H_K$ is a linear operator defined as
\begin{align*}
    (\Omega_\P \phi)(\cdot) & = \int^T_0 K(\psi(X_{0,t},\cdot) \left\langle \phi, \Sigma K(\psi(X_{0,t},\cdot) \right\rangle_{\H_K} dt  \ \ \in \H_K.
\end{align*}
where $\Sigma$ is defined previously at the start of this example. However, since \eqref{eq:3-kernel_inst_variance} is not in the form of $J(\langle \phi, \Phi_X \rangle_{\H_K})$, we cannot directly apply the result of \cite{DeVito2004SomeMethods, Cirone2025RoughHedging} and consequently, we leave this as an avenue for future research.

\section{Numerical Results} \label{sec:6-numerical_results}

\subsection{Path-Dependent Drift} \label{subsec:path-dep-drift}

First, let us consider a one-dimensional setting in which the trader has access to a Markovian signal, but that the signals impact on the drift is non-Markovian and path-dependent. We consider the following setup
\begin{align*}
    dX_t & = \mu_t dt + \sigma_X dW_t \\
    \mu_t & = \gamma \int^t_0 G(t-s) I_s ds \\
    dI_t & = -\kappa I_t dt + \sigma_I dW_t^I.
\end{align*}
where $X$ is the underlying asset process. The stochastic drift $\mu_t$ is a convolution which embeds memory of the past trajectory of the raw signal $I$ via the convolution kernel $G$. For example, we could consider $I_t$ to characterise the order flow of a related asset, as it is shown that cross impact decays slowly according to a power-law \cite{Benzaquen2016DissectingAnalysis,Brokmann2014SlowMarkets}. In this scenario, the order flow $I_t$ would slowly impact the drift of the price, $\mu_t$, over time. That is, suppose we consider a power law kernel
\begin{align} \label{eq:power_law_kernel}
    G(t-s) = \frac{1}{(c+t-s)^\alpha}
\end{align}
where $\alpha >0$ determines the memory of the process such that a smaller $\alpha$ encodes a slower decay and a larger $\alpha$ ensures a faster decay and a higher weighting to more recent information, meanwhile $c>0$ is a small constant that ensures stability for when $t-s \approx 0$. We can see this demonstrated in Figure~\ref{fig:power_law_kernel} for when $\alpha=0.6$.
\begin{figure}[H]
    \centering
    \includegraphics[width=\linewidth]{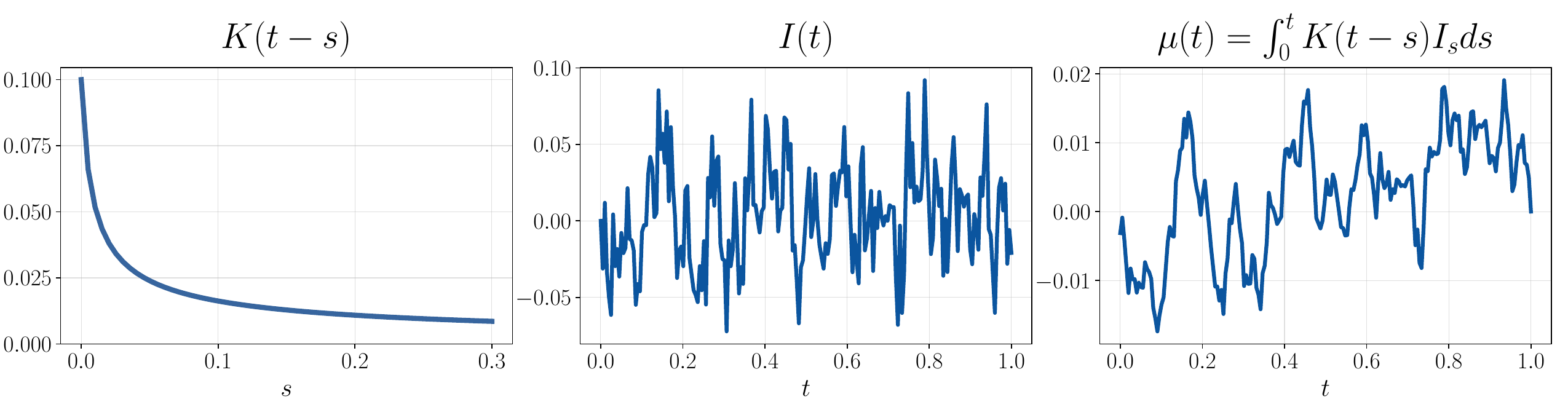}
    \caption{\centering Power law kernel $K(t-s)$ for $c=1,\alpha=0.6$ (LHS). Raw Markovian signal $I_t$ (centre). Drift term $\mu_t$ (RHS).}
    \label{fig:power_law_kernel}
\end{figure}
\vspace{-0.1cm}
Alternatively, we may also have other information, perhaps we already have a prediction $\hat{y}_t$, and that the relationship between our prediction, other factors (e.g. order flow) and the true drift is path-dependent, as seen in this example. Essentially, $I$ represents the information that we have access to, and its relationship with $\mu$ is not instantaneous. In this example, we include a scaling constant $\gamma$ for the stochastic drift $\mu_t$, so that we are able to scale the signal to noise ratio within the system. Suppose that the user wishes to fix $\text{corr}(dX_t, \mu_t dt) = \rho$. Then, for a fixed $\rho, \sigma_X, G, I$, we have
\begin{align*}
    \gamma = \frac{\sigma_X}{\sigma_\mu} \frac{\rho}{\sqrt{1-\rho^2}}
\end{align*}
where $\sigma_\mu$ is the standard deviation of $\mu$ which ensures that the correlation condition is satisfied.

We note that this example is a specific case of the setting explored in Section~\ref{subsec:stoch_drift_markowitz}, in which we have $dX_t = \mu_tdt + dM_t$. We showed that the Markowitz strategy in this setting is not optimal when we consider a terminal variance constraint and that the optimal solution depends on the temporal structure of the stochastic drift term $\mu_t$. Therefore, we wish to understand two key aspects in which the kernel method (and other path-dependent methods) can demonstrate:
\begin{enumerate}
    \item Can such methods \textit{learn} the drift $\mu_t$ while only having access to the Markovian signal $I_t$?
    \item Do these methods utilise the temporal autocorrelation structure in $I_t$ or $\mu_t$ to reduce the variance of $V_T$?
\end{enumerate}
Figure~\ref{fig:power_law_experiment} allows us to answer both of these questions.
In the left panel, we compare the objective outperformance of the Kernel method to that of a simple linear Markovian strategy $\xi_t \sim I_t$, in which we vary the speed of decay $\alpha$ in the power law kernel.
For small values of $\alpha$, this represents a long memory and slow decay, in which we expect greater outperformance of the path-dependent method relative to the Markovian strategy\footnote{This effect is also amplified, since we fix the signal to noise ratio between $\mu$ and $X$, but by increasing the memory of $\mu$, then this suppresses the signal to noise ratio of $I$ and $X$. This is due to the reflexivity $\text{corr}(I_t, dX_t) = \text{corr}(I_t, \mu_t) \times \text{corr}(\mu_t, dX_t)$ and since we keep $\text{corr}(\mu_t, dX_t) = \rho$ fixed, then as $\alpha$ decreases, $ \text{corr}(I_t, \mu_t)$ decreases and hence so does $\text{corr}(I_t, dX_t)$.}, as the kernel is able to better exploit temporal structure in the signal.
\begin{figure}
    \centering
    \includegraphics[width=\linewidth]{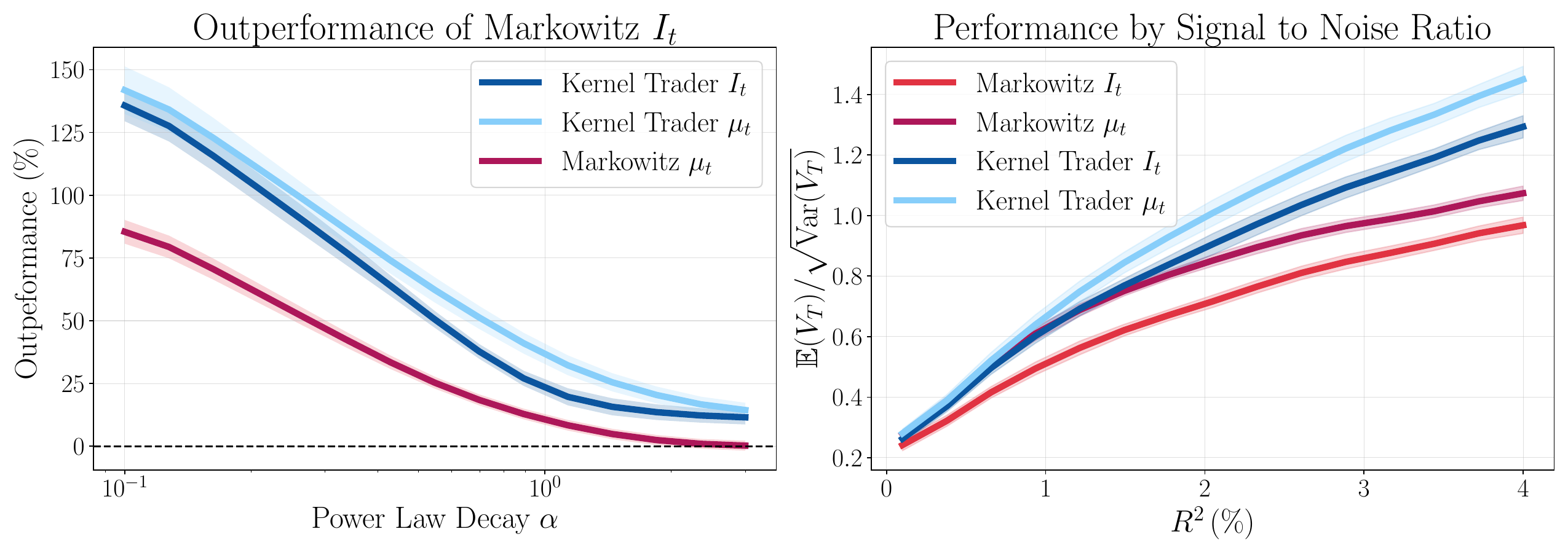}
    \caption{\centering Comparison of the Kernel method and linear Markovian strategy proportional to Markowitz for when we have access to $I$ or $\mu$. The left hand panel displays the relationship with the power law decay $\alpha$ and the right hand panel displays the relationship with the $R^2$ between $\mu$ and $dX$. All results are out of sample with a train size $N=2000$ paths.}
    \label{fig:power_law_experiment}
\end{figure}
\vspace{-0.05cm}

To ensure a fair comparison, we also include the Markowitz strategy that is linear in $\mu_t$, in the scenario where $\mu_t$ is fully observable. However we observe that the outperformance of this strategy is even greater when there is a smaller $\alpha$, suggesting that the path-dependency in $\mu_t$ is being used to a greater effect.
For larger values of $\alpha$, this means that the decay is so fast, it is essentially instantaneous, meaning that $\mu_t \approx I_t$ and therefore the outperformance of Markowitz with $\mu_t$ converges to $0\%$. To answer the first question more explicitly, we can see that the dark blue curve on the left panel is close but not exactly matching that of the light blue curve (corresponding to the kernel trader with access to $\mu_t$). This shows that the kernel method was able to learn a lot of the drift and path-dependencies of $\mu_t$, by only having access to $I_t$, but it was not able to match it entirely.

On the right hand side of Figure~\ref{fig:power_law_experiment}, we instead vary the signal to noise ratio between $\mu_t$ and $dX_t$ for a fixed $\alpha$. We note that as the relationship becomes less noisy, the outperformance becomes more pronounced and in percentage terms, this outperformance is actually fairly constant beyond a certain value.

\subsection{Market Data with Signals}
Previously in Section~\ref{subsec:path-dep-drift}, we constructed an entirely synthetic experiment, so that we could demonstrate how the method performs in a controlled environment under a variety of dynamics, which we were able to via altering the parameters of the model. We are however particularly interested in how the framework performs when using market data, since it exhibits complex, non-trivial structures such as stochastic volatility, intraday effects and regime shifts. However, when working with market data it is often easy to yield spurious results since we cannot control the experiment entirely and so small sample sizes and biases can creep in. Hence, in this example we apply a blend of both, using the market data to provide a non-trivial foundation, but we use a \textit{synthetic} predictive signal, so that we can control parameters such as the $R^2$ of the predictive signal, as well as its decay, forecast horizon, and its temporal structure. Trading algorithms are sensitive to all of the aforementioned stylised facts of a signal, and we wish to test how these impact the performance of the kernel trading framework under a variety of circumstances.
\subsubsection*{Synthetic Signals}
Suppose we are a quantitative trader which has produced a model $\P$, in order to forecast the return of the underlying asset $X$ over the horizon $[t,t+w]$, i.e the predictive signal is given as
\begin{align*}
    \hat{y}_{t,w} = \Ex^\P \left[ \int^{t+w}_t dX_s \, \bigg| \, \F_t \, \right] = \Ex^\P \left[ X_{t+w} \, | \, \F_t \right] - X_t,
\end{align*}
and we denote the true return as $y_{t,w} = X_{t+w} - X_t$. The model $\P$ may originate from a variety of features (e.g. price based, order book based, news based) and a vast array of modelling techniques (from linear regression to neural networks). Therefore, signals are generally non-linear path-dependent functions that retain some of the temporal structure of the prior modelling choices. Hence, in order to model and generate synthetic $\hat{y}_{t,w}$, we put forth the following general model
$$
\hat{y}_{t,w} = \beta_1 \eta_{t,w} + \beta_2 z_{t,w},
$$
where $\beta_1, \beta_2$ are constants, $\eta$ is a stochastic process that represents the \textit{signal} captured, i.e it is a function of the true asset returns $X_{s\in[t,t+w]}$ and $z$ is an uncorrelated \textit{residual noise} process. In practice, we observe \textit{alpha decay}, that is that the true signal decays as
\begin{align*}
    \eta_{t,w} = \int^{w}_0 K(w-s) \, dX_{t+s}
\end{align*}
for some kernel function $K: [0,w] \to \R^d$. 

In practice, unscaled predictions $\hat{y}$ will also likely exhibit the same stylised facts as the true $y$, i.e they will exhibit heavy tails from stochastic volatility and slowly decaying autocorrelation of its absolute values, as well as the same autocorrelation arising from overlapping forecast horizons. Therefore, in order to embed these stylised facts in a synthetic prediction, we must ensure they are represented within the residual noise process $z$. We therefore propose that $z_{t,w}$ be modelled as
\begin{align*}
    z_{t,w} = \int^{t+w}_t d \epsilon_s = 
    \underbrace{ \sqrt{\frac{\gamma}{\mu_V}} \int^{t+w}_t  \sigma_s dW_s^{\mathrm{SV}}}_{\substack{\text{Proportion } \gamma \text{ of Variance} \\ \text{ of Residual is SV Noise}}} + \underbrace{ \sqrt{1-\gamma} \int^{t+w}_t  dW_s^{\mathrm{add}}}_{\substack{\text{Proportion } (1-\gamma) \text{ of Variance } \\ \text{of Residual is Additive Noise}}}
\end{align*}
where $\sigma_t$ is a measure of volatility, $W^{\mathrm{add}}$ and $W^{\mathrm{SV}}$ are independent Brownian motions and $\mu_V =  \Ex[V_t] = \Ex[\sigma_t^2]$ is the expected variance, to ensure each noise component retains variance 1. We don't actually know what the \textit{true} volatility process is at any given time, but we can approximate it in two different ways, which become a better representation as the number of observations increases and time grid becomes finer. Suppose we observe $X_t = X_{t_0},\dots, X_{t_N} = X_{t+w}$, then we can approximate the instantaneous volatility either as the volatility observed over the whole period or as a proxy for instantaneous volatility using the absolute return
\begin{align*}
    1. \quad & \sigma_s \approx \sqrt{ \frac{1}{w} \sum_{t_i \in [t , t+w]} r_{t_i}^2} \quad \forall s \in [t,t+w] \\
    2. \quad & \sigma_s \approx \frac{|r_{t_k}|}{t_{k+1} - t_k} \text{ for } s \in [t_k, t_{k+1}]
\end{align*}
where $r_{t_k}$ is given as $X_{t_{k+1}} - X_{t_{k}}$ for $k = 1, \dots, N$. 
\begin{figure}[H]
    \centering
    \includegraphics[width=\linewidth]{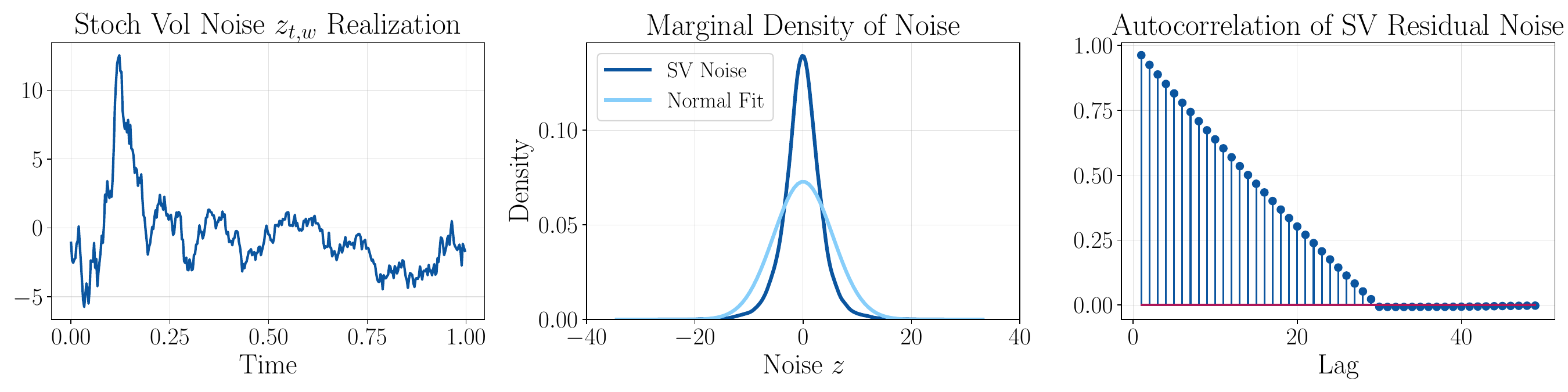}
    \caption{\centering Visualisation of the stochastic volatility noise component which captures the linear autocorrelation of the lookforward horizon $w=30$, and captures a non-normal marginal density with heavy tails. Here we have $\sigma_z = \sqrt{w} = \sqrt{30}$.}
    \label{fig:SV_noise}
\end{figure}
Therefore, this form of residual noise is able to capture heteroskedastic and heavy-tailed residuals through the parameter $\gamma \in [0,1]$, such that when $\gamma=0$, the predictions $\hat{y}$ will look normally distributed for small $R^2$ and when $\gamma=1$ they will mirror the marginal distribution of $y$.

Now, piecing the \textit{signal} and \textit{residual noise} components together, we obtain the final model 
\begin{align} \label{eq:correlated_admissible}
    \notag \hat{y}_{t,w} & = \beta_1 \eta_{t,w} + \beta_2 z_{t,w} \\
    \hat{y}_{t,w} & = \rho \sigma_y \left( \frac{\rho}{\rho_{y\eta}} \frac{1}{\sigma_\eta} \eta_{t,w} + \frac{1}{\sigma_z} \sqrt{1- \frac{\rho^2}{\rho_{y\eta}^2} } z_{t,w} \right)
\end{align}
which we define as a \textit{correlated admissible prediction}. This prediction has the following properties:
\begin{enumerate}
\vspace{-0.1cm}
\itemsep-0.65em
    \item Corr$(y, \hat{y}) = \rho$
    \item $\hat{y}$ minimises the \textit{mean squared error} between itself and the true $y$,
    $$\hat{y} = \argmin_{\eta} \text{Var}(y-\eta).$$
    \vspace{-0.5cm}
    \item The $R^2$ between $\hat{y}$ and the true $y$ is equal to $\rho^2$, i.e
    \begin{align*}
        R^2 = 1 - \frac{\mathrm{Var}(y - \hat{y})}{\mathrm{Var}(y)} = \frac{2 \mathrm{Cov}(y, \hat{y}) - \mathrm{Var}(\hat{y})}{\mathrm{Var}(y)} = \frac{\mathrm{Var}(\hat{y})}{\mathrm{Var}(y)} = \rho^2.
    \end{align*}
    \vspace{-0.5cm}
\end{enumerate}
Therefore, we now have a model that is reflective of the complexities arising in market data, but we are able to control 
\begin{itemize}
\vspace{-0.2cm}
\itemsep-0.5em
    \item The correlation and $R^2$ of the predictive signal
    \item The speed and type of alpha decay (e.g. exponential or power-law)
    \item The original prediction horizon $w$
    \item How heavy-tailed and heteroskedastic the residual errors are via $\gamma$.
\end{itemize}

\subsubsection*{The Experiment}
In this experiment we obtain\footnote{We obtained this data from \hyperlink{Databento}{www.databento.com}, \cite{Databento2025NasdaqMSFT.}.} intraday 5-minutely bars of Microsoft (MSFT) and set up a trading strategy during the regular session of trading between 09:30-16:00 EST. This results in 78 distinct trading times per day. Throughout, we train on the period of 01/01/2020 - 31/12/2023 and subsequently use 01/01/2024 - 31/03/2025 as our testing period which is the latest available at the time of writing. This results in a total of $\sim$78k trades in sample and $\sim$20k trades out of sample. A main advantage of having a synthetic signal is that we are able to choose specifically the amount of signal to noise in our system. For a low $R^2$, it is often possible to obtain spurious results due to small sample sizes in which the residual noise process dominates --- even when in the order of tens of thousands. In order to eliminate (or at least quantify) this randomness, we simulate from many random seeds in which to generate synthetic predictions (as in Section~\ref{subsec:path-dep-drift}) to obtain not only an expectation of performance but also the distribution in performance due to randomness arising from the predictions.

Throughout, we consider a 15-minute forecast horizon $w$ with an $R^2$ of $0.5\%$, the reader may believe this is too high or too low, however this was an arbitrary choice and also chosen to be small enough to demonstrate performance in low signal to noise ratio circumstances. We also assume that the decay kernel $K$ is given as a power law which is the same as seen in \eqref{eq:power_law_kernel} for a given power law decay $\alpha$.
\begin{figure}[h]
    \centering
    \includegraphics[width=\linewidth]{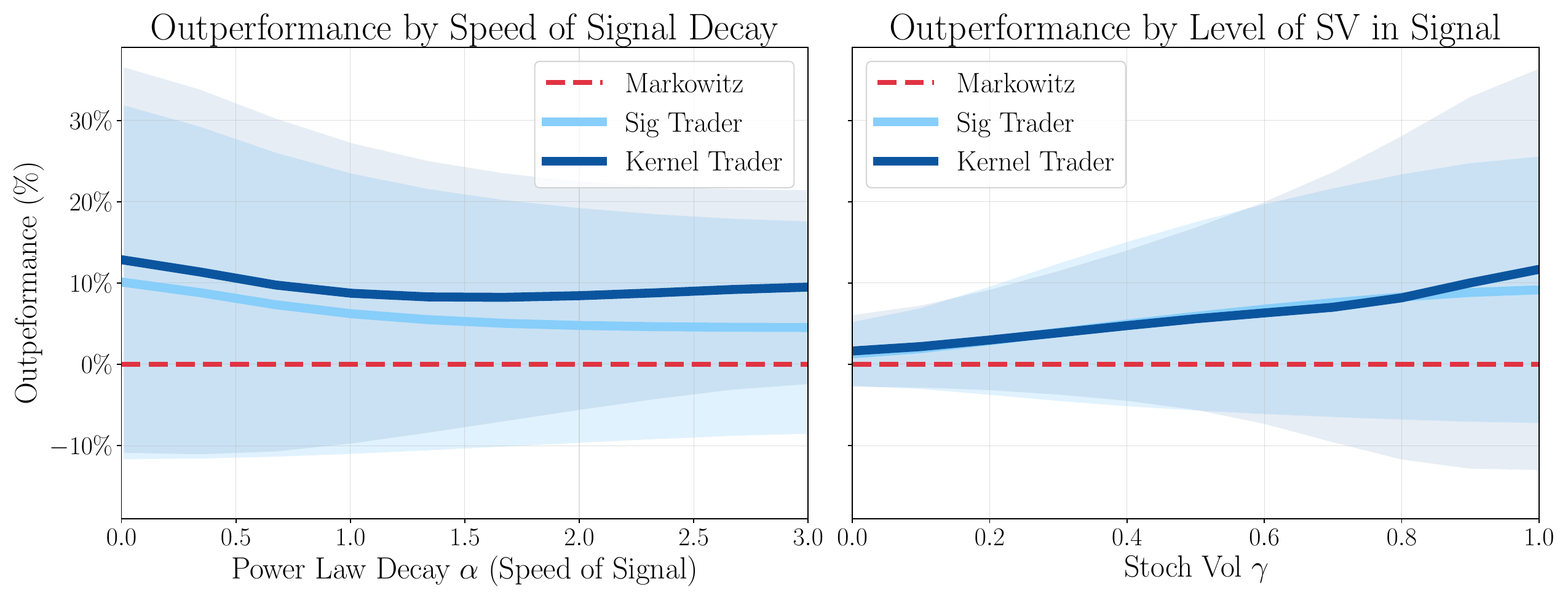}
    \caption{\centering The left hand panel demonstrates the outperformance in the objective of the Markowitz strategy for varying speeds of signal for when the decay kernel $K$ is a power law and $\gamma=1$. The right hand panel shows the outperformance when the stochastic volatility (and heavy tail) parameter $\gamma$ is varied between 0 and 1, for when the decay $\alpha=0.5$. The Sig-Trader is truncated at order $K=5$. All results are out of sample on the test of 01/01/2024 - 31/03/2025 ($\sim$20k trades).}
    \label{fig:market-data}
\end{figure}
The results in Figure~\ref{fig:market-data} demonstrate a consistent outperformance of the Markovian (Markowitz) strategy that is linear in the predictive signal $\hat{y}$.  The left-hand panel shows that this outperformance holds regardless of whether the signal decays quickly or slowly, with a marginal advantage when the signal is slower and exhibits greater path-dependency. Notably, we also observe that the uncertainty (and the discrepancy between path-dependent methods and Markovian methods) increases when the signal has longer memory, which is reflected in the wider confidence bands around the mean.

Similarly, the right-hand panel shows consistent outperformance as the synthetic signals become increasingly heavy-tailed arising from stochastic volatility. This higlights how, even with the same underlying asset and the same predictive performance in $R^2$, if two signals have different temporal structures (e.g. in the form of stochastic volatility) then this can cause very different results. Not only are the results different in expectation, but the uncertainty inherited is much greater for when $\gamma$ is larger. Finally, we observe just how similar the performances are of the Kernel Trader and the Sig-Trader (\cite{Futter2023SignatureSignalsb}), not only in their expected outperformance but also the distribution of performance for different random seeds in which to generate the synthetic signals. 
It is encouraging that both of the path-dependent methods are similar, meaning they may be close to the “true” optimum performance.
We wanted to ensure that all results are fair, realistic and comparable and so these results were ran for 30 simulations (random seeds), which provides the shaded band around each solid line. We note that numerical results are usually sensitive to the exact testing period, the length of period and also the choice of regularisation parameters that have been optimised, but by simulating over many random seeds we are able to smooth out any irregularities to obtain more accurate results. While the synthetic signals provide a neat structure in which to test in, in reality predictive signals can be more complex and dependent on other sources of information such as order flow and price movements in related assets. Therefore, if one were implementing these methods in practice, it may be more beneficial to include these factors in the feature embedding $\psi$ and this may result in an even greater outperformance.

\section{Comparison with Signature Trading \cite{Futter2023SignatureSignalsb}} \label{sec:4-comparison_sig_trading}

The objective criterion that is considered in Section~\ref{sec:3-mean_variance} is the same objective that is explored in \cite{Futter2023SignatureSignalsb} in which the authors parameterised the trading strategy as a linear functional on the signature. It is natural then for us to compare the kernel method solution derived in Section~\ref{sec:3-mean_variance}, for when the chosen kernel in our framework is the \textit{Signature Kernel}. The main questions we aim to explore are
\vspace{-0.2cm}
\begin{itemize}
\itemsep-0.5em
    \item Are kernel methods more efficient and robust with smaller datasets?
    \item Are kernel methods more computationally efficient with a large number of assets/signals?
    \item Does the signature method converge to the kernel method as we increase the order of truncation?
\end{itemize}

\subsection*{Mathematical Comparison}

The position $\xi_t^{\text{ker}}$ of a kernel trading strategy is given as
\begin{equation*}
\xi_t^{\text{ker}} = \phi^*(\psi(X_{0,t})) =   \underbrace{ \frac{1}{N \lambda} \Gamma_{\P}\left(\psi(X_{0,t}) \right)}_{\in \R^{d \times N}} \underbrace {\left( \lambda \text{Id}_N + \frac{\eta}{N} \K_{\Phi}(\bfX, \bfX)(\textup{Id}_N - \frac{1}{N} \boldsymbol{1}_N \boldsymbol{1}_N^\top) \right)^{\dag}  \boldsymbol{1}_N}_{\text{Optimal Weights } \alpha^* \in \R^N},
\end{equation*}
where $K_\Phi$ is induced by the \textit{Signature kernel}. Meanwhile, the position $\xi_t^{\text{sig}}$ of a signature trading strategy is given as
\begin{align*}
    \xi_t^{\text{sig}} = \phi^*(\psi(X_{0,t})) = \underbrace{\text{Sig}(\psi(X_{0,t}))}_{\in \R^M} \underbrace{\left( \lambda \text{Id}_{dM} \, + \, \eta \Sigma^{\text{sig}}(\bfX)  \right)^{-1} \mu^{\text{sig}}(\bfX) }_{\text{Optimal Linear Functionals } \ell^* \, \in  \, \R^{M \times d}}
\end{align*}
where $M = \lvert \W^K_{N+d+1} \rvert$ is the number of terms in the truncated signature of $\psi(X_{0,t})$ at order $K$. Both methods rely on the same feature maps to extract path-dependencies and non-linearities, namely $\psi$ and $\text{Sig}(\cdot)$. However, in the signature case, the dynamic component is universal for all assets, while there is a different set of weights (linear functions) for each asset. In contrast, the kernel method learns a single set of universal weights $\alpha^*$, since the asset dependence is embedded directly into the feature map $\Gamma_{\mathbb{P}}$.

It is essential to note that the Sig-Trader method must \textit{truncate} the signature at order $K$, since we cannot compute the full (infinite) signature. Meanwhile, the signature kernel method does not require truncation, enabling it to capture higher-order terms that may improve strategy performance.

\subsection*{Learning Markovian Dynamics} \label{sec:4-learning_markov}

Let us now consider the case in which each asset $X^m$ is given as an OU process for $m = 1, \dots, d$, i.e.,
\begin{align*}
    dX_t^m = \underbrace{\theta_m (\mu_m - X_t^m)}_{\mu_t} dt + \underbrace{\sigma_m dW_t^m}_{dM_t}, \ \ \ \forall m=1,\dots,d,
\end{align*}
where $W = (W^1, \dots, W^d)$ is a $d$-dimensional correlated Brownian motion, and $\theta_m$, $\mu_m$, $\sigma_m$ are constants.
 We note that this example is a specific case of the stochastic drift scenario introduced in Section~\ref{subsec:stoch_drift_markowitz}, where we have $\mu_t = \theta (\mu - X_t)$ and $dM_t = \sigma dW_t$. Therefore, the \textit{Markowitz} strategy is given as
\begin{align*}
    \xi_t^M \sim \Sigma^{-1} (\mu - X_t).
\end{align*}
up to some scaling constant. It is particularly relevant here to understand whether both the signature and kernel methods can learn or outperform the Markowitz solution, and how quickly they do so in practice.

\subsection{Sample Size Convergence \& Path Length}

We are interested to explore whether kernel methods can be more efficient at learning than the original feature map itself when given a smaller number of training samples.
However, this can be a delicate scenario to test and is very much a data-dependent problem, since the complexity of the relationships being learnt and the signal-to-noise ratio can vary the results. Therefore, in relation to this, the results will also depend on the regularisation used, the choice of feature map and the choice of $m<N$ eigenvectors in the robust $\alpha$ setting of Section~\ref{subsec:3-robust_alpha}. 
In Figure~\ref{fig:regularisation_kernel}, we show how the optimal $\lambda^*$ out of sample changes as we increase the sample size. Therefore, for a fair comparison in this context, when we fit using a given sample size $N$, we optimise the regularisation parameter $\lambda^*$ and record the objective in this case.

\begin{figure}[h]
\centering
  \includegraphics[width=\linewidth]{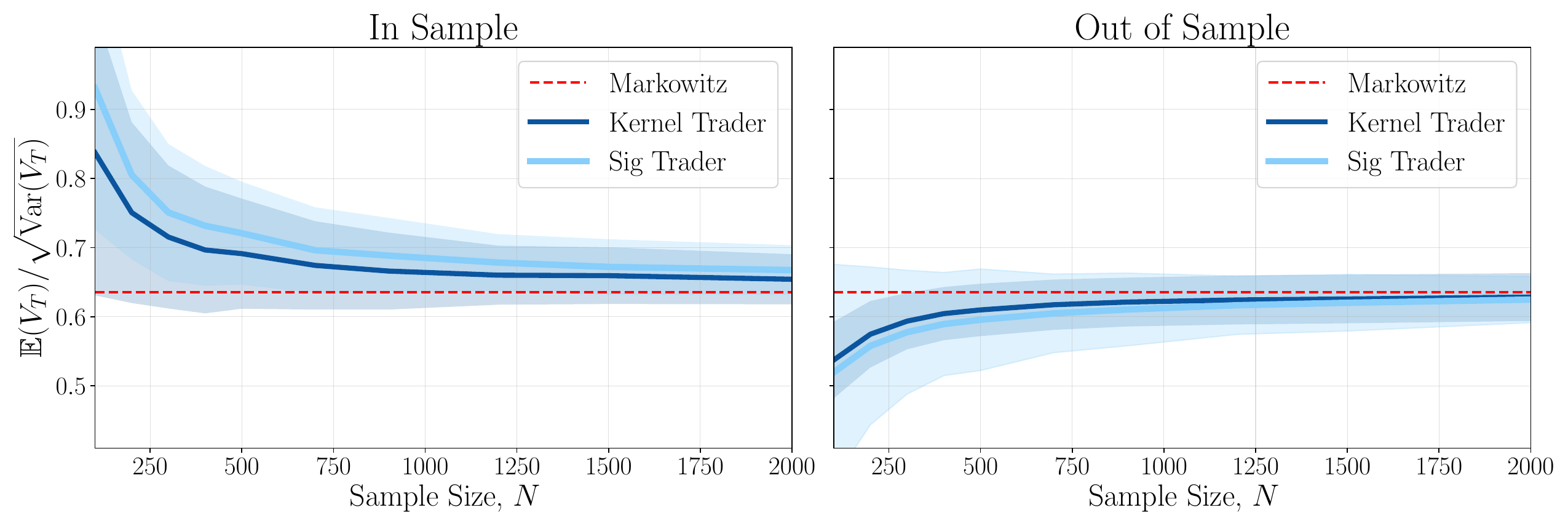}
  \caption{\centering Sample size convergence when learning from the underlying asset. The length of paths is $20$. The LHS shows the in sample performance as the number of the number of training samples increases, while the RHS shows the corresponding out of sample convergence.}
  \label{fig:sample_convergence}
\end{figure}

Figure~\ref{fig:sample_convergence} shows that the convergence in performance is relatively comparable between the two methods once regularised. We do note that the kernel method converges slightly faster, and that its out of sample variance for smaller sample sets is also lower, suggesting it to be more \textit{robust} to out of sample uncertainty.
In this scenario, both methods are capable of learning simple Markovian dynamics directly from the underlying process, without requiring a helping hand in the form of a signal or an expectation of the drift. It is also able to learn the inverse covariance matrix weightings directly. 
\begin{figure}[H]
    \centering
    \includegraphics[width=\linewidth]{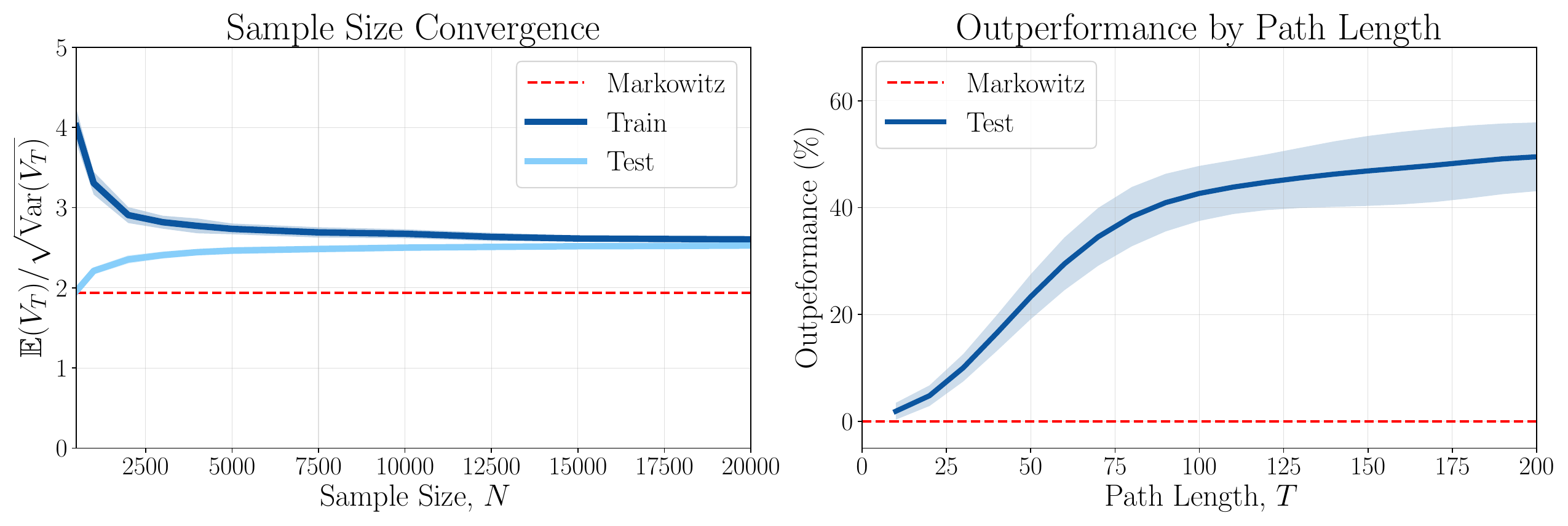}
    \caption{\centering Sample size convergence for when we have paths of length 50 (LHS). The RHS panel shows how the \textit{out of sample} outperformance of the Markowitz strategy increases as the path length increases, when trained on $N=5000$ paths.}
    \label{fig:sig_samples_pathlength}
\end{figure}
\vspace{-0.2cm}
In this example we compare in the case where the Sig-Trader is truncated at order 4 and the path length\footnote{In the experiments we opted for $T=20$, as this choice ensures that one could efficiently run many iterations for all sample sizes in a feasible time, as the signature kernel computation typically scales as $\mathcal{O}(T^2)$.\\ Truncation at order $4$ is a computationally motivated choice, while aiming to include reasonably many signature terms. Higher order terms often tend to be more influenced by noise in measurement error.} is $T=20$. The purpose of this experiment is to demonstrate the convergence for small sample sizes, however both methods do not (significantly) outperform the Markowitz strategy (out of sample) for $N=2000$ samples, while in Section~\ref{subsec:stoch_drift_markowitz} we discuss how there should be an outperformance since the \textit{Markowitz} solution becomes suboptimal. This is down to two main things, the first of which is that in fact to truly achieve the best performance possible (such as that \textit{in sample}), then one requires at least $N=2000$ samples. Secondly, the path length in this case $(T=20)$ is simply not long enough to produce a large effect. Since the true optimal solution requires learning the covariance structure of the drift along the whole path, the magnitude of this and hence the difference between Markowitz and the path-dependent methods is small. Therefore, the left hand panel of Figure~\ref{fig:sig_samples_pathlength} demonstrates for when we have paths of length $50$, there is an obvious outperformance of the Markowitz strategy.
In order to highlight this further, the right hand panel illustrates that for longer paths, there is a greater outperformance of the Markowitz strategy in $\%$ terms (out of sample). 
\begin{figure}[H]
\centering
  \includegraphics[width=\linewidth]{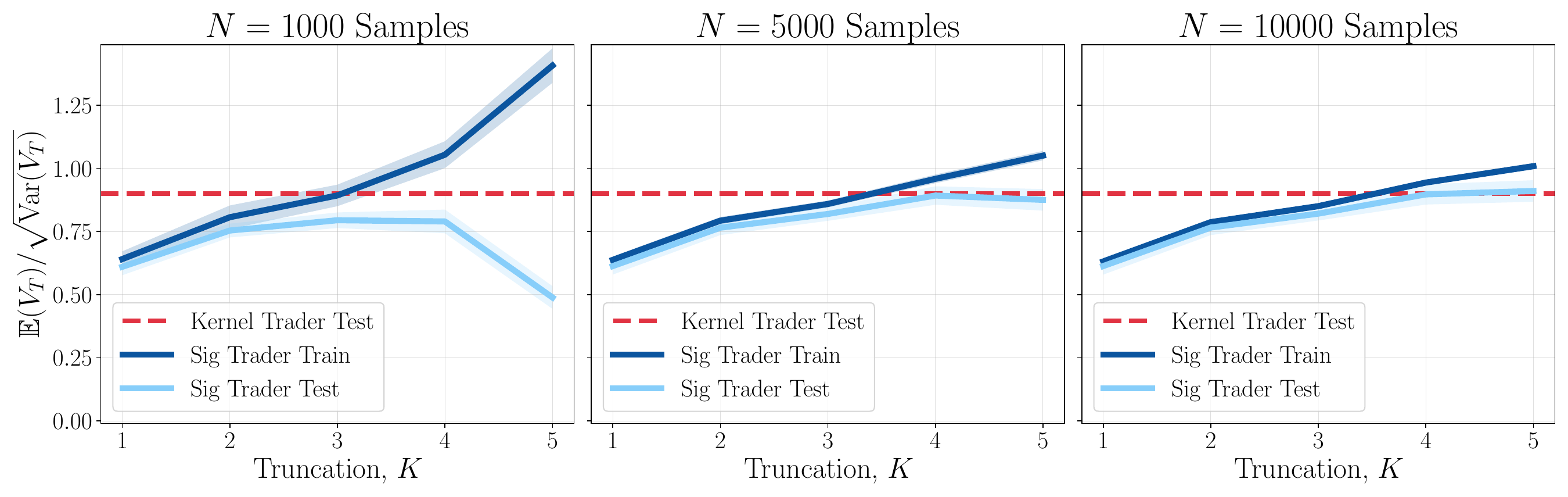}
  \caption{\centering Comparison of the Signature method (without any regularisation) for higher levels of truncation and samples. The kernel method is trained on $N=2000$ samples by comparison.}
  \label{fig:truncation}
\end{figure}
\vspace{-0.4cm}
We also note that while in the experiments above, we have similar performance between the kernel method and the signature kernel method, this is mainly due to the signature being truncated at order 5, and that this is a good enough approximation for the complexity of the relationships being learnt. This is reflected in Figure~\ref{fig:truncation} which shows how the kernel method is an upper bound for the signature method, in which the performance of the Sig-Trader converges to that of the Kernel trader as the truncation $K$ increases. In this experiment, no regularisation is used for the Sig-Trader in order to demonstrate its dependence on sample size.

\vspace{-0.2cm}
\subsection{Computational Complexity}

The main distinction between these two methods, and between kernel methods more generally, is the computational efficiency. Since we are not able to compute the entire untruncated signature in practice, we instead compute the inner product between two signatures using the \textit{kernel trick} proposed in \cite{Salvi2020ThePDE}. Therefore, we wish to explore when each method is more computationally efficient, depending on
\begin{itemize}
\vspace{-0.2cm}
\itemsep-0.5em
    \item The dimension $d$ of the input feature space $\psi(X_{0,t})$
    \item The length of the path $T$
    \item The number of training samples $N$
    \item The order of truncation $K$ in the Sig-Trader framework.
\end{itemize}
\vspace{-0.2cm}
The kernel trick method of \cite{Salvi2020ThePDE} requires solving a linear hyperbolic PDE to compute the signature kernel, which is non-linear in the path length $T$, typically scaling as $\mathcal{O}(T^2)$. Since we are computing a Gram matrix of $N \times N$ kernel evaluations, this scales quadratically in $N$. As $N$ grows, we must also invert a larger matrix, though kernel evaluation often dominates the inversion cost for large $N$. We compare results using the \href{https://github.com/crispitagorico/sigkernel}{\texttt{sigkernel}} package, with a dyadic order equal to zero.
\begin{figure}
\centering
  \includegraphics[width=\linewidth]{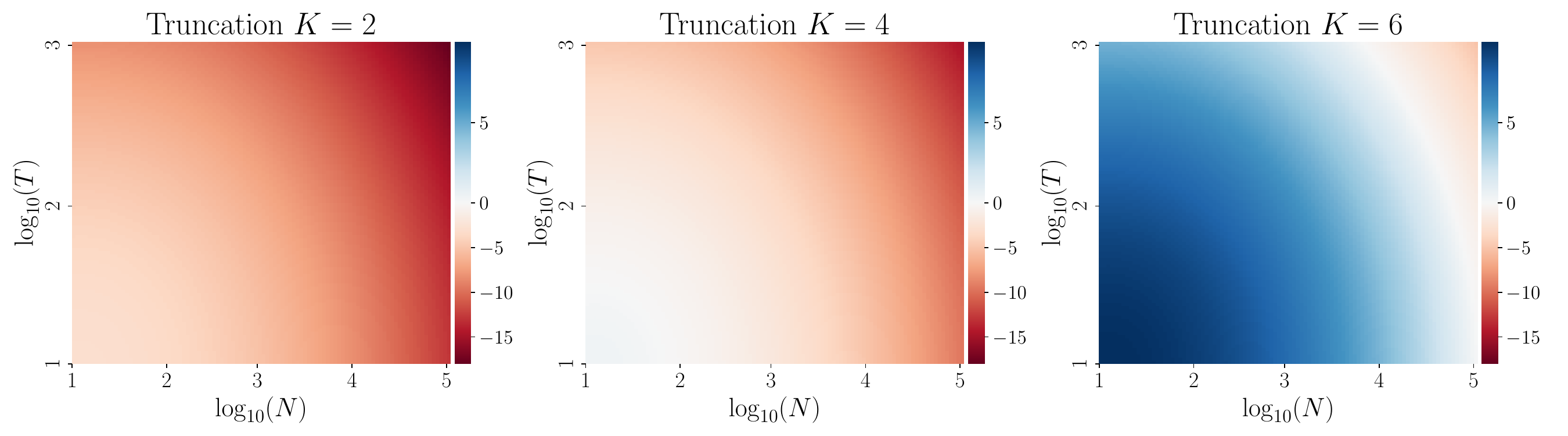}
  \caption{\centering The heatmaps above show the relative computational runtime of the Sig-Trader vs Kernel Trader. The colorbar values represent the difference in log-time of the computation.
  }
  \label{fig:computation_dim}
\end{figure}

Figure~\ref{fig:computation_dim} illustrates the difference in (log) computational runtime of the kernel trader framework and the truncated signature framework for different orders of truncation $K$.  A negative (red) value indicates that the Kernel Trader takes longer to compute, while blue indicates that it is more efficient. In this example, we fix the dimension of the inputs as $d=4$. We observe that for longer paths and more samples, the signature method is faster (red areas) for truncation $K \leq 4$, however for larger orders of truncation (e.g when $K=6$), the signature kernel becomes faster to compute across most $N,T$.
Likewise, in Figure~\ref{fig:computation_truncation}, we compare how the input dimension (i.e. number of channels) affects the computational runtime. In this example, we fix the order of truncation for the signature at $K=4$ and notice that the inflection point occurs around $d=5$ input channels. For example, when we have paths of length 100 and ~1000 samples, we observe very similar runtimes, represented by the white circular region in the middle ($d=5$) heatmap.
\begin{figure}[H]
\centering
  \includegraphics[width=\linewidth]{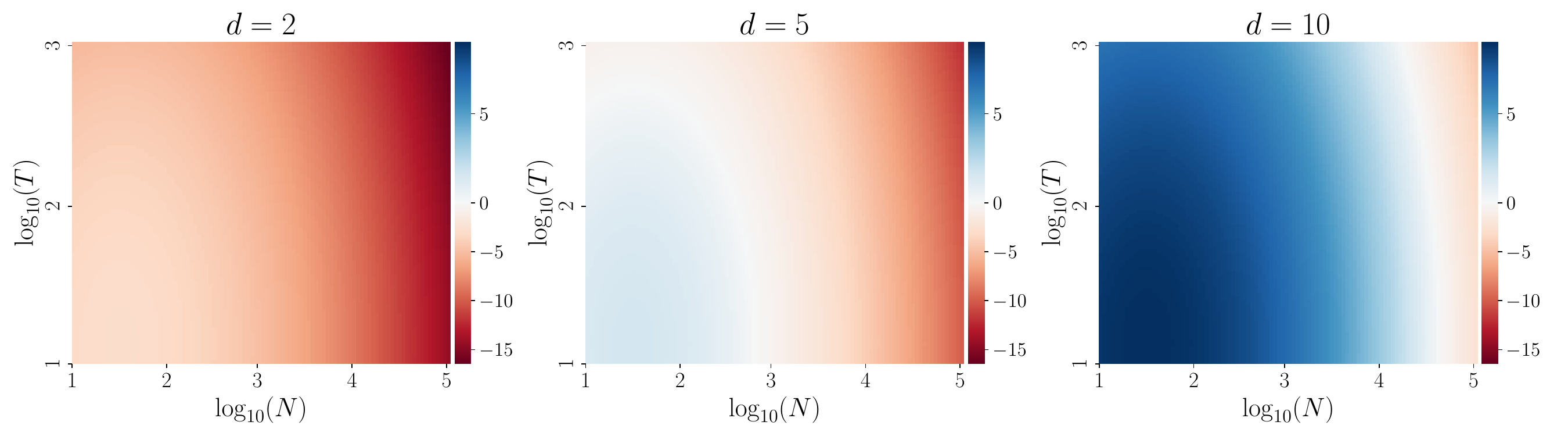}
  \caption{\centering The heatmaps show the relative computational demand of the Sig-Trader vs Kernel Trader for when the number of channels of the input paths is equal to $d=2,5,10$.}
  \label{fig:computation_truncation}
\end{figure}
These comparisons are based on runtime\footnote{All runtimes were recorded on two NVIDIA GeForce RTX 2080 Ti GPUs.} only and do not account for memory usage. The batch size for Gram matrix computation was chosen to minimise runtime while remaining within GPU memory limits. With access to more GPUs or parallel compute, the kernel method could scale more efficiently, although we did not benchmark this explicitly. We also note that kernel methods tend to use more memory due to the PDE-based computation of the signature kernel, which can limit batch size and introduce bottlenecks.

An alternative implementation, \href{https://github.com/tgcsaba/KSig}{\texttt{KSig}}, was recently released to support GPU-accelerated computation of the signature kernel~\cite{Toth2025AKernel}, and is compatible with \texttt{Scikit-Learn}. Although we did not benchmark this implementation directly, our aim was to highlight the trade-off over the parameter space $(N, T, K, d)$. There exists a region in this space where each method is more efficient, and this boundary may simply shift with alternative or optimised implementations of the signature kernel.

\subsection{Summary of Comparison} \label{subsec:5.3-comparison}

\begin{table}[H]
\centering
\renewcommand{\arraystretch}{1.4}
\begin{tabularx}{\textwidth}{
  >{\raggedright\arraybackslash}m{3cm}  
  >{\raggedright\arraybackslash}X       
  >{\raggedright\arraybackslash}X       
}
\toprule
\textbf{ } & \textbf{Signature Trading \cite{Futter2023SignatureSignalsb}} & \textbf{Kernel Trading} \\
\midrule
\vspace{0.35cm}
\textbf{\makecell[l]{\textbf{Strategy}\\\textbf{Performance}} } 
& Higher orders of truncation $K$ allow for greater model complexity but require regularisation to avoid overfitting. If it outperforms the kernel method, this may reflect the kernel model not being fully optimised or regularised.
& Often higher accuracy due to richer feature space but overfitting must be mitigated with implicit regularisation through scaling, explicit regularisation through $\lambda$ and spectral decomposition via SVD. \\

\textbf{Computation} 
& Linear (or at least sub-quadratic) in path length $O(T)$ and sample size $O(N)$, but exponential in truncation order $O(d^K)$. Truncation is flexible. No Gram matrix is required, and fitting is parallelisable. More efficient for \textit{online} computation since only the signature of the current online trajectory requires to be computed.
& Quadratic in both path length and sample size $O(T^2 N^2)$ using PDE solvers. Parallelisable across batches, but memory-intensive for long paths. More efficient for high dimensional input channels $d$. Less efficient for \textit{online} computation due to kernel evaluations performed against all co-location paths. \\

\vspace{0.35cm}
\textbf{\makecell[l]{\textbf{Flexibility \&}\\\textbf{Interpretability}}} 
& Input features $\psi(X_{0,t})$ are typically constrained to lower dimensionality due to computational cost. Signature terms are more interpretable and intuitive than kernel features.
& Flexible choice of kernels (not limited to the signature kernel). Alternative feature maps (e.g. RBF) can also be used. Allows for larger and richer feature representations, but less interpretable. \\

\textbf{Practicalities} 
& Much less sensitive to path scaling. The covariance matrix is generally better conditioned than kernel Gram matrices and typically do not require spectral decomposition. Hyperparameter optimisation is simpler, with fewer parameters to tune.
& Sensitive to path scaling and must be optimised. Often requires spectral decomposition, and selection of effective rank $m < N$. More general and flexible, but introduces a larger, more complex hyperparameter space. \\

\bottomrule
\end{tabularx}
\label{tab:sig-vs-kernel}
\end{table}

In our experiments, we find that the signature method and the kernel method are broadly comparable in performance (especially at higher orders of truncation), once practicalities are accounted for and hyperparameters are well-optimised. However, this trade-off is likely to be highly data-dependent. We discuss further in Section~\ref{sec:5-implementation} how to maximise the out-of-sample performance of the kernel method, in particular how to balance expressivity, generalisation, and the risk of under/overfitting. We summarise the key differences between the two approaches in the table above.

In summary, the kernel framework provides a more general and rich modelling framework for capturing complex relationships within the data. However, this comes at the cost of greater engineering complexity and sensitivity to hyperparameter choices. For a single asset or smaller-scale setting, the signature method may offer a more efficient and interpretable alternative. For larger multi-asset problems or richer input structures, the kernel method offers greater scalability.
However, for high-frequency tasks which are sensitive to online computation and latency constraints, the Sig-Trading method is better suited. 
\vspace{-0.2cm}
\section{Implementation \& Practicalities} \label{sec:5-implementation}
\vspace{-0.1cm}
Kernel methods, when equipped with non-linear feature maps, are particularly powerful due to being a more lightweight alternative to deep learning tools, especially since the solution that we obtain is analytic and does not require gradient-based optimisation. However, while the theoretical framework of kernel methods does provide a powerful foundation, we do still require some careful modelling choices when implementing the method in practice to ensure that its potential is maximised. In particular, we are especially concerned with under/overfitting to ensure generalisability out of sample, and this is dependent on two hyperparameters $\lambda, \gamma$. Additionally, to ensure robustness of performance out of sample we discuss how one may find a “cleaner” solution using singular value decomposition (SVD). Finally, we also consider the computational aspect of the method and discuss how to make this more efficient using parallelisation and path sub-sampling.

\subsection*{Regularisation Hyperparameter $\lambda$}
We previously showed in \eqref{eq:3-mean-var_alpha_star} that the optimal weights $\alpha^*$ are given as
    \begin{align} 
        \alpha^*(\bfX) = \left( \lambda \text{Id}_N + \frac{\eta}{N} \K_{\Phi}(\bfX, \bfX)\left(\textup{Id}_N - \frac{1}{N} \boldsymbol{1}_N \boldsymbol{1}_N^\top \right) \right)^{\dag}  \boldsymbol{1}_N,
    \end{align}
where the regularisation parameter $\lambda$ must be tuned as a hyperparameter, as this is dependent on the complexity of the dynamics that are being learnt, sample size and scaling of the data. 
\vspace{-0.2cm}
\begin{figure}[H]
\centering
  \includegraphics[width=\linewidth]{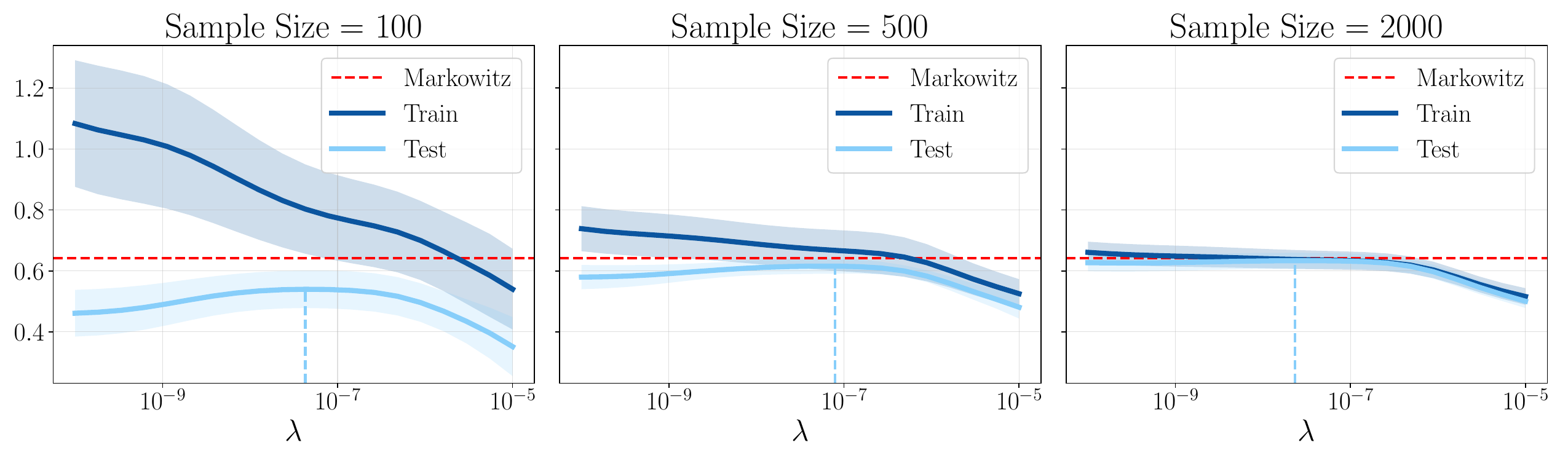}
  \caption{\centering The performance objective (both in and out of sample) for varying levels of regularisation $\lambda$. This is shown for $N=100,500,2000$ training samples from left to right.}
  \label{fig:regularisation_kernel}
\end{figure}
\vspace{-0.5cm}
If we choose a regularisation that is too small, then we will overfit, and if we choose a regularisation parameter too large then we may underfit, hence $\lambda$ must be carefully tuned to balance model complexity and generalisation.  
As illustrated in Figure~\ref{fig:regularisation_kernel}, the sensitivity to $\lambda$ is particularly pronounced for smaller datasets. In practice, $\lambda$ is typically chosen via \textit{cross-validation}, by maximising performance on a validation set using a preferred method (e.g. walk-forward). Fortunately, this hyperparameter is relatively cheap to optimise for since it does not require to re-compute the features or gram matrix. In the original formulation (non-spectral) then the just simply requires to re-invert a matrix, while in the spectral case we require to re-compute the spectral decomposition for each $\lambda$, both of which operations are relatively cheap in comparison to computing the gram matrix.

\subsection*{Path Scaling Hyperparameter $\gamma$} \label{subsec:5-scaling}

An important consideration in machine learning models (particularly those based on time series data) is how to scale the input channels. In signature methods, scaling the input path by a constant $\gamma \in \R$ leads to a homogeneous scaling of the signature terms: each term of order $k$ scales as $\gamma^k$. This follows from the algebraic structure of the signature transform, which respects the tensor algebra $T((\R^d))$. As a result, the objective in the signature trading framework is theoretically scale-invariant, up to floating point precision. However, when the input signal is scaled too small, higher-order signature terms (which scale as $\gamma^k$) may vanish due to underflow and loss of precision, effectively collapsing the representation. This phenomenon is visible on the RHS of Figure~\ref{fig:scaling}, where model performance degrades for very small $\gamma$. However, the performance remains stable outside of this regime for when we do not run into floating point errors.
\begin{figure}[H]
\centering
  \includegraphics[width=\linewidth]{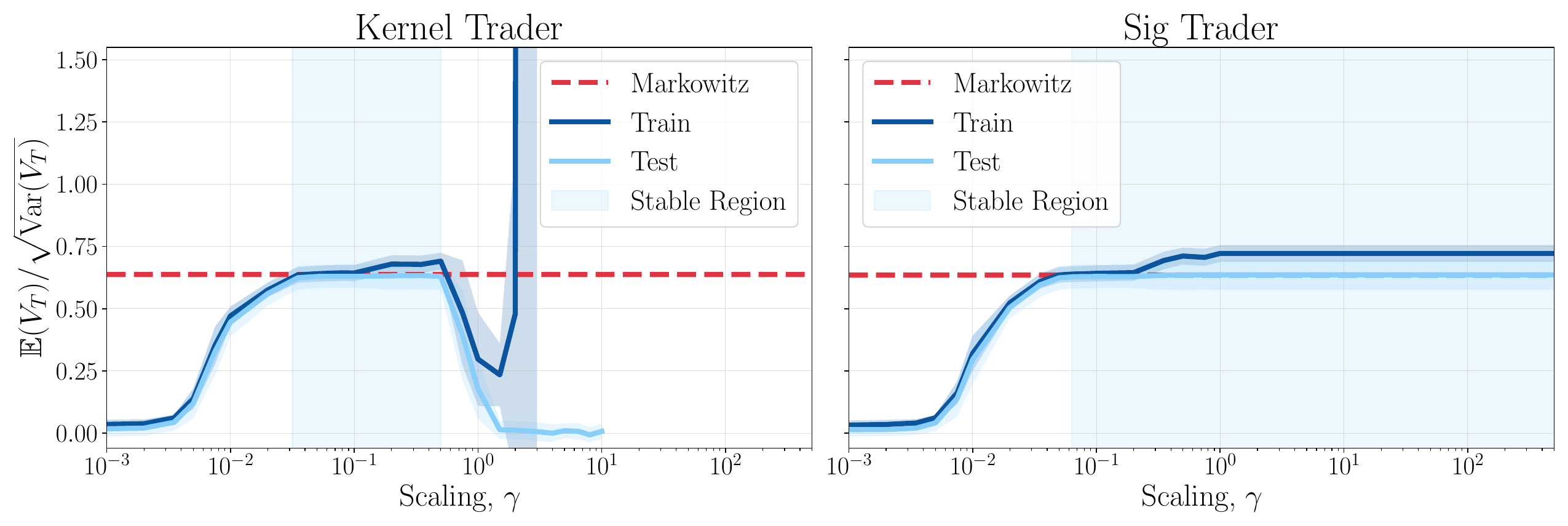}
  \caption{\centering The effect of the scaling hyperparameter $\gamma$ on in and out of sample performance. On the LHS we observe the Kernel Trader performance and the RHS displays the corresponding Sig-Trader performance.}
  \label{fig:scaling}
\end{figure}
Likewise, the scaling issue is also prevalent in kernel methods in order to extract the right amount of complexity. Especially in the signature kernel setting, if we scale the input paths $X, Y$ by a constant $\gamma \in \R$, then we obtain
\begin{align*}
    K_{\text{sig}}(\gamma X,\gamma Y) & = \langle \text{Sig}(\gamma X), \text{Sig}(\gamma Y)\rangle \\
    & = \sum_{\bluebf{w}\in \W} \text{Sig}^{\bluebf{w}}(\gamma X) \text{Sig}^{\bluebf{w}}(\gamma Y)
    \\
    & = \sum_{k=0}^\infty \sum_{\bluebf{w} \in \W_k} \gamma^{2k} \text{Sig}^{\bluebf{w}}(X) \text{Sig}^{\bluebf{w}}(Y).
\end{align*}
Therefore, we observe that higher order terms (larger $k$) are scaled disproportionately. By setting $\gamma$ to be smaller, the higher order terms \textit{vanish} due to the higher powers of $\gamma^{2k}$ and the lower order signature terms therefore dominate. In this sense $\gamma$ acts an implicit regularisation parameter since the higher order signature terms contain more nonlinear information. Meanwhile, increasing the scaling parameter tends to overfit more to higher order terms, until the model breaks down and the signature kernel does not converge anymore, which we observe in the LHS of Figure~\ref{fig:scaling}. 

Hence, there is a balance to strike and the scaling parameter $\gamma$ becomes a hyperparameter that needs to be tuned in accordance to the $L^2$ regularisation parameter $\lambda$. As discussed for $\lambda$, the scaling parameter $\gamma$ can also be tuned using cross-validation or otherwise.

\subsection*{Variance Constraint via $\eta$}
In practice, we may want a given level of variance at the terminal time, i.e that is we want to set the constraint $\text{Var}_\P^X(V_T^\xi(X)) = \Delta$. In this case, the variance is controlled through the scaling of the weights \(\alpha\) via the risk-aversion parameter \(\eta\).

We note that
\begin{align*}
    \mu_\Phi & = \K_{\Phi}(\bfX, \bfX) \mathbf{1}_M \\
    \Sigma_\Phi & = \K_{\Phi}(\bfX, \bfX) \Xi_\P^\Phi(\bfX).
\end{align*}
We require the constraint 
\begin{align*}
    \alpha^\top \Sigma_\Phi \alpha = \alpha^\top (\K_{\Phi}(\bfX, \bfX) \Xi_\P^\Phi(\bfX)) \alpha = \Delta.
\end{align*}
Since we know that
\begin{align*}
     \alpha = \left( \lambda \text{Id}_N + \frac{\eta}{N}\Xi_\P^\Phi(\bfX) \right)^{-1}  \boldsymbol{1}_N.
\end{align*}
then we obtain
\begin{align*}
    \Delta = \left( \left( \lambda \text{Id}_N + \frac{\eta}{N}\Xi_\P^\Phi(\bfX) \right)^{-1}  \boldsymbol{1}_N \right)^\top  (\K_{\Phi}(\bfX, \bfX) \Xi_\P^\Phi(\bfX)) \left( \left( \lambda \text{Id}_N + \frac{\eta}{N}\Xi_\P^\Phi(\bfX) \right)^{-1}  \boldsymbol{1}_N \right).
\end{align*}
which we now need to solve to obtain $\eta$ for a desired $\Delta$. However, it is not possible to solve analytically for \(\eta\) in closed form due to the additional $\lambda \text{Id}$ term, however without this, the solution becomes unstable. Therefore, in order to find $\eta$ for a given $\Delta$, we require a numerical optimisation via any user-chosen solver.
Interestingly, solving numerically reveals a power-law relationship between \(\lambda\) and \(\eta\), however this may depend on the structure of the specific gram matrix being used.

For example, on the RHS of Figure~\ref{fig:risk_aversion}, setting $\sqrt{\Delta} = 0.05$, by changing $\lambda$ (between $10^{-4}$ and $10^{-6}$, when solving for the optimal $\eta$, we obtain a power law relationship, that is 
\begin{align*}
    \eta^*(\Delta) = a(\Delta) \lambda^{p(\Delta)}.
\end{align*}
In the RHS of Figure~\ref{fig:risk_aversion}, we obtain $a(\Delta)=5$, $p(\Delta) = -0.83$ so that they are almost perfectly inversely proportional, however we note that of course this relationship will vary based on $K_\Phi$, meaning $a$ and $p$ will change, they are also dependent on $\Delta$.

\begin{figure}[H]
    \centering
    \includegraphics[width=0.95\linewidth]{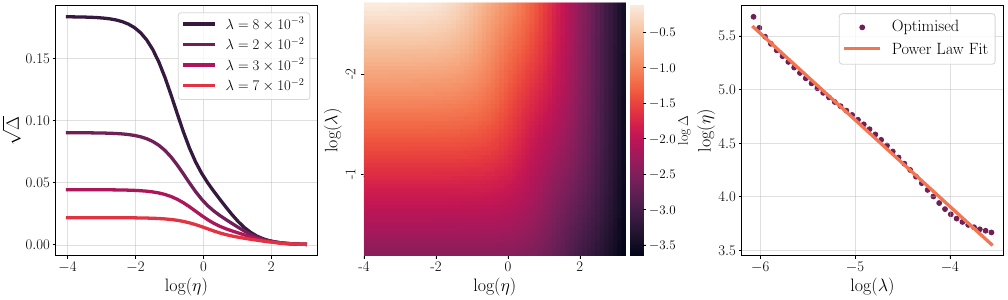}
    \caption{\centering The LHS plot shows how the relationship between the standard deviation of PnL ($\sqrt{\Delta}$) is non-linear in $\eta$. The central plot displays $\log\Delta$ as a function of $\log \lambda$ and $\log \eta$. The RHS shows for a fixed $\sqrt{\Delta}=0.05$, the relationship between $\lambda$ and the given $\eta^*(\Delta)$.}
    \label{fig:risk_aversion}
\end{figure}

\subsection*{Spectral Representation of $\alpha^*$}

In Section~\ref{subsec:3-robust_alpha}, we recast $\alpha$ in terms of its eigenvectors and eigenvalues. Previously, this required inverting a matrix that is often nearly singular, leading to instability, especially for small values of $\lambda$. As a result, attempts to regularise or optimise $\lambda^*$ via cross-validation often yielded unstable solutions. To address this instability, we truncate the summation to the top \(m < N\) eigenvalues and eigenvectors, aiming for greater numerical stability. The solution for $\alpha$ is then given as
\begin{align*}
\alpha^*_m = \frac{\sqrt{N}}{\lambda} \left( \sum_{k=1}^m \frac{(\mathbf{u}_k^\top \mathbf{e}_N)^2}{\gamma_k} \right)^{-1} \sum_{k=1}^m \frac{\mathbf{u}_k^\top \mathbf{e}_N}{\gamma_k} \mathbf{u}_k.
\end{align*}

We can see this explicitly in Figure~\ref{fig:svd_noisiness}, where for the left and centre panels, the objective ratio becomes unstable for smaller values of $\lambda$. In this example, we have simulated $N=2000$ sample co-location paths and have tested on $N=2000$ test paths. When $m=2000$, we recover the full solution of \eqref{eq:3-mean-var_alpha_star}, which is clearly unstable for small $\lambda$, however as $m$ decreases, the curves become much more stable and smooth.

\begin{figure}
    \centering
    \includegraphics[width=\linewidth]{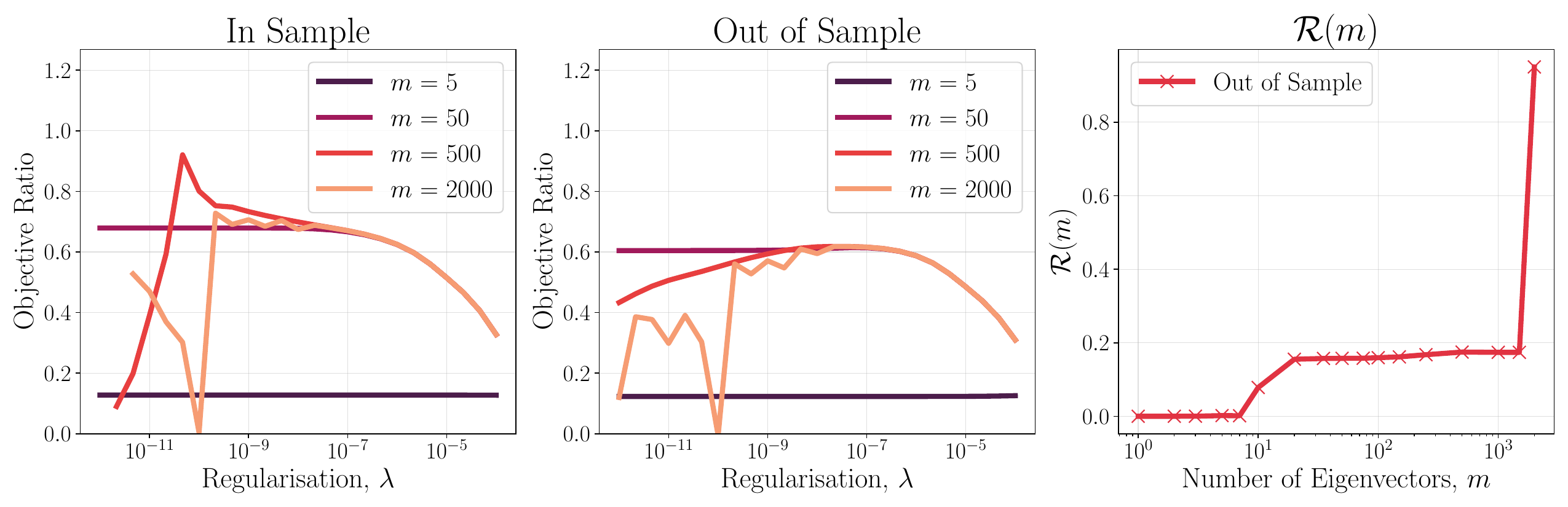}
    \caption{\centering Visualising the instability of $\alpha^*$ for different values of regularisation $\lambda$ (LHS and Centre). The RHS displays the metric $\mathcal{R}(m)$ defined above as a proxy for instability.}
    \label{fig:svd_noisiness}
\end{figure}

To demonstrate this, we can create a metric $\mathcal{R}$. Firstly, let us define the criterion ratio $J$,
\begin{align*}
    J( \alpha \, |  \, \lambda, m ) =  \frac{\alpha_{\lambda,m}^\top \mu_{\Phi}}{\sqrt{\alpha_{\lambda,m}^\top \Sigma_\Phi \, \alpha_{\lambda,m}}} 
\end{align*}
that measures the expected return for a given level of risk. This is an invariant criterion such that if we double the values of $\alpha$ (proportionally) then the ratio will stay the same. Now, we define the metric $\mathcal{R}$ as
\begin{align*}
    \mathcal{R}(m) = \left\lVert \frac{\partial J(\lambda \, | \, m)}{\partial \log \lambda} \right\rVert_{L^2} = \sqrt{\sum_{i \in \I} \frac{J(\lambda_{i+1} \, | \, m) - J(\lambda_i \, | \, m)}{\log(\lambda_{i+1}) - \log(\lambda_i)}}
\end{align*}
where $m$ is the number of eigenvectors corresponding to the top $m$ eigenvalues. This metric is a smoothness penalty for the solutions $\alpha_m$ which admit such noisy behaviour across values of $\lambda$. This is demonstrated further on the right panel of Figure~\ref{fig:svd_noisiness}, where for smaller $m$, we get less noisy results. However, we can see from the centre figure that in fact by removing too many eigenvectors (small $m$, e.g $m=5$), we also remove almost all of the signal from the system, which results in both the in and out of sample performance becoming much worse. Thus, this suggests that an optimal $m$ exists, which in this case is observed to be in the region $m\in[50,500]$.

\subsection*{Mean-Variance Implementation}

\vspace{0.5cm}
\begin{algorithm}[H]
\caption{Fitting the Optimal Kernel Trading Strategy}
\label{alg:algo1}
\hspace*{\algorithmicindent} \textbf{Input:} Market paths $\bfX = \{X^i_{t \in [0,T]}\}_{i=1}^N \subset \Lambda_X^\alpha \subset C^\alpha([0,T], \mathbb{R}^d)$ \\
\hspace*{\algorithmicindent} \hspace*{3.5em} Feature embedding $\psi: \Lambda_X^\alpha \to  \Lambda^\psi$ \\
\hspace*{\algorithmicindent} \hspace*{3.5em} Kernel $k_\varphi(X^i, X^j) := \langle \varphi(\psi(X^i)), \varphi(\psi(X^j)) \rangle_{\mathcal{H}_K}$ \\
\hspace*{\algorithmicindent} \hspace*{3.5em} Hyperparameter grids: $\Gamma$, $\Lambda$, $M$ for $\gamma$, $\lambda$, and $m$ respectively \\
\hspace*{\algorithmicindent} \hspace*{3.5em} Validation objective $\mathcal{L}(\alpha^*, \text{val})$
\begin{algorithmic}[1]
    \For{$\gamma \in \Gamma$}
        \State Scale input paths: $\tilde{X}^i \gets \gamma X^i$
        \State Compute kernel $\Phi$ gram matrix $\K_\Phi(\tilde{\bfX}, \tilde{\bfX})$ where $(\K_{\Phi})_{ij} \gets k_\varphi(\tilde{X}^i, \tilde{X}^j)$
        \For{$\lambda \in \Lambda$}
            \State Form regularized operator: $A \gets \lambda \text{Id}_N + \frac{\eta}{N} \K_\Phi(\tilde{\bfX}, \tilde{\bfX})$
            \State \parbox[t]{\dimexpr\linewidth-\algorithmicindent}{%
                Compute eigendecomposition: 
                \newline \hspace*{\algorithmicindent} \hspace*{5.5em} where  $A = U \Theta U^\top$, $U = (\mathbf{u}_1, \dots, \mathbf{u}_N)$ and $\Theta = \text{diag}(\theta_1, \dots, \theta_N)$.
            }
            \For{$m \in M$}
                \State Compute weights:
                \[
                \alpha^*_m = \frac{\sqrt{N}}{\lambda} \left( \sum_{k=1}^m \frac{(\mathbf{u}_k^\top \mathbf{e}_N)^2}{\theta_k} \right)^{-1} \sum_{k=1}^m \frac{\mathbf{u}_k^\top \mathbf{e}_N}{\theta_k} \mathbf{u}_k
                \]
                \State Evaluate performance: $\mathcal{L}(\alpha^*_m)$
                \State Store $(\gamma, \lambda, m)$ if objective improves
            \EndFor
        \EndFor
    \EndFor
    \State \Return best $(\gamma, \lambda, m)$ and corresponding $\alpha^*_m$
\end{algorithmic}
\end{algorithm}
Throughout this work, we particularly focus on the \textit{signature kernel} and use the package \href{https://github.com/crispitagorico/sigkernel}{\texttt{sigkernel}} (\cite{Salvi2020ThePDE}), alongside \texttt{PyTorch} and \texttt{JAX} for calculating and performing functionality related to tensors and the signature kernel. \href{https://github.com/tgcsaba/KSig}{\texttt{KSig}}  \cite{Toth2025AKernelb} is another package that offers signature kernel computation. However we do stress that this work can be used with any path-dependent kernels with respect to any feature map such as DTW and alignment kernels \cite{Shimodaira2001DynamicMachine,Cuturi2007AAlignments.} or convolutional and sequential kernels \cite{Haussler1999ConvolutionStructures}. 

All that is required in order to fit the kernel trading strategy optimal weights $\alpha^*$ is to compute the \textit{“PnL feature map”} and its corresponding gram matrix $\K_\Phi$ for a given kernel through time. How one chooses to construct the gram matrix of co-location points is up to the user, and will be dependent on the complexity of the relationship being learnt. This provides a great deal of flexibility, even when working with time series and signatures. For example, we may also define a randomised signature kernel as $k_{\text{rSig}}(X,Y) = \langle \psi(X), \psi(Y)\rangle$ where $\psi(X) = \varphi(\textup{r-Sig}(X))$ for some activation function $\varphi$.
The idea here is that instead of computing the high-dimensional object, we project the signature into a random lower-dimensional space and then apply non-linearity using the activation function. We provide the algorithm to compute an optimal trading strategy in Algorithm \ref{alg:algo1}. Likewise, the algorithm for computing the trading strategy online is shown in Algorithm~\ref{alg:algo2}, as also illustraded in Figure~\ref{fig:online_phase}.

\vspace{0.2cm}
\begin{algorithm}[ht]
\caption{Online Kernel Trading Strategy at time $t \in [0,T]$}
\label{alg:algo2}
\hspace*{\algorithmicindent} \textbf{Input:} Live asset path $X_{0,t}$ \\
\hspace*{\algorithmicindent} \hspace*{3.2em} Trained weights $\alpha^* \in \mathbb{R}^N$ from Algorithm~\ref{alg:algo1} \\
\hspace*{\algorithmicindent} \hspace*{3.2em} Historical asset paths $\bfX = \{X^i_{t \in [0,T]}\}_{i=1}^N \subset \Lambda_X^\alpha$ \\
\hspace*{\algorithmicindent} \hspace*{3.2em} Feature embedding $\psi: \Lambda_X^\alpha \to 
\Lambda^\psi$ \\
\hspace*{\algorithmicindent} \hspace*{3.2em} Kernel $K: \Lambda^\psi \times  \Lambda^\psi \to \mathcal{L}(\R^d)$

\hspace*{\algorithmicindent} \textbf{Output:} Positions $\xi_t^m$ for each asset $m \in \{1, \dots, d\}$
\begin{algorithmic}[1]
    \For{each time $t$ in trading horizon}
        \State Compute current feature path: $\psi(X_{0,t})$
        \For{each asset $m = 1, \dots, d$}
            \State Compute the \textit{PnL feature map}:
            \vspace{-0.2cm}
            $$
            \left[\Gamma_{\mathbb{P}}(\psi(X_{0,t}))\right]_{m,i} = \int_0^T K(\psi(X_{0,t}), \psi(X^i_{0,s})) e_m \, dX^i_s
            $$
        \EndFor
        \State Form $\Gamma_{\mathbb{P}}(\psi(X_{0,t})) \in \mathbb{R}^{d \times N}$
        \State Compute trading position:
        \vspace{-0.2cm}
        $$
        \xi_t = \Gamma_{\mathbb{P}}(\psi(X_{0,t})) \cdot \alpha^* \, \in \R^d
        $$
        \vspace{-0.2cm}
    \EndFor
    \State \Return positions $\xi_t = (\xi_t^1, \dots, \xi_t^d)$ for all $t$
\end{algorithmic}
\end{algorithm}
\vspace{0.2cm}

\subsection*{Parallelisation and Path Sub-sampling}

As discussed in Section~\ref{sec:4-comparison_sig_trading}, the computation of the kernel Gram matrix $\mathcal{K}_\Phi(\mathbf{X}, \mathbf{X})$ scales quadratically with the number of sample paths $N$, requiring $\mathcal{O}(N^2)$ kernel evaluations, dominating runtime during model fitting. Likewise for the online stage, this requires $O(N)$ kernel evaluations with the colocation trajectories. Due to memory constraints, the Gram matrix is often computed in batches, and this process can be easily parallelised.

An alternative strategy is to sub-sample the set of training trajectories to a much smaller subset of size $n \ll N$. Let $\bfX = \{X^i\}_{i=1}^N$ denote the full set of market trajectories. We define the index set $\mathcal{I} := \{1, \ldots, n\}$ and the complementary set $\mathcal{J} := \{n+1, \ldots, N\}$. Define the corresponding “landmark” subset of paths $\mathbf{X}_{\mathcal{I}} := \{X^i\}_{i \in \mathcal{I}}$. We can then partition the gram matrix into a block matrix as
\begin{align*}
    \K_\Phi(\bfX, \bfX) = 
    \begin{pmatrix}
        \K_\Phi(\bfX_{\I}, \bfX_{\I}) & \K_\Phi(\bfX_{\I}, \bfX_{\J}) \\
        \K_\Phi(\bfX_{\J}, \bfX_{\I}) & \K_\Phi(\bfX_{\J}, \bfX_{\J})
    \end{pmatrix}
    \in \R^{N\times N}.
\end{align*}
Then, to create a low-rank approximation of the gram matrix $\K_\Phi(\bfX, \bfX)$, we can use the Nystr\"{o}m approximation \cite{Williams2000UsingMachines, Drineas2005OnLearning}. Let
$$
C := \mathcal{K}_\Phi(\mathbf{X}, \mathbf{X}_{\mathcal{I}}) \in \mathbb{R}^{N \times n}, \quad
W := \mathcal{K}_\Phi(\mathbf{X}_{\mathcal{I}}, \mathbf{X}_{\mathcal{I}}) \in \mathbb{R}^{n \times n}.
$$
Then the low-rank approximation is given by:
\begin{align*}
\tilde{\K}_\Phi := C W^\dag C^\top \in \mathbb{R}^{N \times N},
\end{align*}
where $W^\dag$ denotes the Moore--Penrose pseudo-inverse of $W$. This approximation preserves symmetry and positive semi-definiteness, and
the approximation $\tilde{\K}$ lies in the span of the landmark columns. Its rank is at most $n$.
Therefore, instead of computing $N^2$ kernel evaulations, we are able to compute only $N \times n$ evaluations which can be much more efficient for smaller $n$.
The selection of the landmark index set $\mathcal{I}$ can be random, but performance may improve with more informed choices. For example, one may prefer higher variance paths, paths that are more “recent” or rather a mix of paths over time.
Nyström approximations are part of the Sketching literature~\cite{chen2021accumulationsprojectionsaunifiedframework, pmlr-v139-song21c}, where strategies to reduce the computational cost of Kernel methods are extensively studied.

\section{Conclusion}

In this work, we introduced a kernel-based framework for constructing dynamic, path-dependent trading strategies under a mean-variance optimisation criterion. Building on the theoretical foundations of reproducing kernel Hilbert spaces (RKHS), we demonstrated how trading strategies can be parameterised as a function in an RKHS, allowing for a flexible and general approach to modelling non-Markovian dependencies in financial data. This framework supports a wide range of modelling choices and provides a closed-form alternative to gradient-based deep learning methods.
We also derived a robust spectral decomposition of the optimal solution, which reduces sensitivity to hyperparameter selection and improves out-of-sample stability.

Through both synthetic and market-data experiments, we illustrated that kernel trading strategies consistently outperform their Markovian counterparts, especially in settings where asset dynamics or predictive signals exhibit temporal structure and non-Markovianity. Furthermore, we conducted a detailed comparison with the signature trading approach of \cite{Futter2023SignatureSignalsb}, highlighting comparable performance when using the signature kernel. However, the kernel framework offers greater modelling flexibility and scalability, particularly when dealing with complex relationships or a large number of assets and signals.
We also discussed implementation and practical considerations in depth, including path scaling, variance constraints, regularisation choices, and computational efficiency.
This kernel-based approach offers a versatile framework for financial optimisation problems, and while we only studied mean-variance optimisation in this work, we lay foundations for future research including objectives with market impact and transaction costs.

\clearpage

\begin{appendices}
\addtocontents{toc}{\protect\setcounter{tocdepth}{-1}}
\section{Rough Path Theory \& Kernel Preliminaries} \label{sec:appx_rough_paths}

\begin{definition}[Positive Definite (Scalar) Kernel]\label{def:positive-kernel}
A symmetric function \( K : \mathcal{X} \times \mathcal{X} \rightarrow \mathbb{R} \) is called a \emph{positive definite kernel} if for any finite set \( \mathbf{X} = \{x_1, \ldots, x_n\} \subset \mathcal{X} \) and any \( c_1, \ldots, c_n \in \mathbb{R} \), it holds that
\[
\sum_{i=1}^{n} \sum_{j=1}^{n} c_i c_j K(x_i, x_j) \geq 0.
\]
The corresponding kernel matrix \( K(\mathbf{X}, \mathbf{X}) = [K(x_i, x_j)]_{i,j=1}^n \) is symmetric and positive semi-definite.
\end{definition}

\begin{definition}[Reproducing Kernel Hilbert Space (RKHS)]\label{def:rkhs}
A Hilbert space \( \mathcal{H}_K \) of functions \( f: \mathcal{X} \rightarrow \mathbb{R} \) is called a \emph{reproducing kernel Hilbert space} if there exists a symmetric positive definite kernel \( K : \mathcal{X} \times \mathcal{X} \rightarrow \mathbb{R} \) such that:
\begin{enumerate}
    \item[(i)] For all \( x \in \mathcal{X} \), \( K(x, \cdot) \in \mathcal{H}_K \),
    \item[(ii)] For all \( f \in \mathcal{H}_K \) and all \( x \in \mathcal{X} \), the \emph{reproducing property} holds:
    \[
    f(x) = \langle f, K(x, \cdot) \rangle_{\mathcal{H}_K}.
    \]
\end{enumerate}
If we define the canonical feature map \( \psi(x) := K(x, \cdot) \), then
\[
K(x, x') = \langle \psi(x), \psi(x') \rangle_{\mathcal{H}_K}.
\]
\end{definition}

\begin{remark}
Given a map \( \phi: \mathcal{X} \rightarrow \mathcal{H} \) into a Hilbert space \( \mathcal{H} \), the associated kernel is
\[
K(x, y) := \langle \phi(x), \phi(y) \rangle_{\mathcal{H}}.
\]
This defines an RKHS where functions \( f \) take the form \( f(x) = \langle w, \phi(x) \rangle_{\mathcal{H}} \) for some \( w \in \mathcal{H} \).
\end{remark}

\begin{definition}[Operator-Valued Kernel]\label{def:operator-kernel}
Given a set \( \mathcal{Z} \) and a real Hilbert space \( \mathcal{Y} \), an \( \mathcal{L}(\mathcal{Y}) \)-valued kernel is a map
\[
K : \mathcal{Z} \times \mathcal{Z} \rightarrow \mathcal{L}(\mathcal{Y}),
\]
such that:
\begin{enumerate}
    \item For all \( z, z' \in \mathcal{Z} \), \( K(z, z') = K(z', z)^* \) (adjoint),
    \item For any finite set \( \{(z_i, y_i)\}_{i=1}^N \subset \mathcal{Z} \times \mathcal{Y} \), the matrix
    \[
    \left[\langle K(z_i, z_j) y_i, y_j \rangle_{\mathcal{Y}}\right]_{i,j=1}^N
    \]
    is positive semi-definite.
\end{enumerate}
\end{definition}

\begin{example}[Separable Operator-Valued Kernel]
Let \( \mathcal{Y} = \mathbb{R}^d \), and let \( k : \mathcal{Z} \times \mathcal{Z} \rightarrow \mathbb{R} \) be a scalar-valued positive definite kernel (e.g., Gaussian, polynomial, etc.). Define:
\[
K(z, z') := k(z, z') \cdot I_d,
\]
where \( I_d \in \mathcal{L}(\mathbb{R}^d) \) is the identity operator on \( \mathbb{R}^d \). Then \( K \) is an operator-valued kernel because:
\begin{enumerate}
    \item \( K(z, z') = K(z', z)^* \), since \( k(z, z') = k(z', z) \) and \( I_d \) is self-adjoint,
    \item For any \( \{(z_i, y_i)\}_{i=1}^N \subset \mathcal{Z} \times \mathbb{R}^d \),
    \[
    \sum_{i,j} \langle K(z_i, z_j) y_i, y_j \rangle = \sum_{i,j} k(z_i, z_j) \langle y_i, y_j \rangle,
    \]
    which defines a positive semi-definite Gram matrix as a Kronecker product: \( K_{\text{ovk}} = k(\mathbf{Z}, \mathbf{Z}) \otimes I_d \succeq 0 \).
\end{enumerate}
\end{example}

\begin{definition}[$p$-Loss] \label{defn:p-loss}
Let $\mathcal{X}$ be a locally compact second countable space and $\mathcal{Y}$ a closed subspace of $\R$.
Given \( p \in [1, +\infty) \), a function \( V : \mathbb{R} \times \mathcal{Y} \to [0,+\infty) \) is called a \emph{$p$-loss function} with respect to a probability measure \( \mu \) on \( \mathcal{X} \times \mathcal{Y} \) if:
\begin{enumerate}
    \item For all \( y \in \mathcal{Y} \), \( V(\cdot, y): \mathbb{R} \to [0,+\infty) \) is convex,
    \item \( V \) is $\mu$-measurable,
    \item There exist \( b \geq 0 \) and a measurable function \( a: \mathcal{Y} \to [0,+\infty) \) such that
    \[
    V(w,y) \leq a(y) + b|w|^\alpha \quad \forall (w,y) \in \mathbb{R} \times \mathcal{Y}, \qquad \int_{\mathcal{X} \times \mathcal{Y}} a(y) d\mu(x,y) < +\infty.
    \]
\end{enumerate}
\end{definition}

\begin{example}[Squared Error as a \( p \)-Loss]
Define the function $V: \R \times \R \to [0, +\infty)$ as
\[
V(x, y) = (x - y)^2.
\]
Then \( V \) satisfies the conditions of a \( p \)-loss (Definition~\ref{defn:p-loss}) with \( p = 2 \) with respect to the any measure $\mu$ absolutely continuous with respect to Lebesgue measure, with finite $y$-marginal second moment. Specifically:
\begin{enumerate}
    \item \( V(\cdot, y) \) is convex for all \( y \in \R \),
    \item \( V \) is continuous and hence measurable,
    \item Setting \( a(y) := 2|y|^2 \), \( b := 2 \), and \( \alpha = 2 \), we have
    \[
    (x - y)^2 \leq 2x^2 + 2y^2 = b|x|^2 + a(y),
    \]
    and the integrability condition \( \int a(y) \, d\mu(x, y) < \infty \) thanks to finiteness of the second moment.
\end{enumerate}
Thus, squared error is a valid \( p \)-loss with \( p = 2 \).
\end{example}

\begin{theorem}[General RKHS Functional Minimizer \cite{DeVito2004SomeMethods}] \label{thm:rkhs-min}
Let \( \mathbb{P} \) be a probability measure on \( \mathcal{X} \times \mathcal{Y} \), and let \( \mathcal{H}_K \) be an RKHS with kernel \( K \) bounded w.r.t. \( \mathbb{P} \). For a \( p \)-loss function \( V \) and \( \lambda > 0 \), define:
\[
f^* \in \arg \min_{f \in \mathcal{H}_K} \left\{ \mathbb{E}_{(x, y) \sim \mathbb{P}}[V(y, f(x))] + \lambda \|f\|_{\mathcal{H}_K}^2 \right\}.
\]
Then \( f^* \) satisfies:
\[
f^*(x) = -\frac{1}{2\lambda} \mathbb{E}_{(x', y') \sim \mathbb{P}} \left[ \alpha^*(x', y') K(x, x') \right],
\]
where \( \alpha^*(x, y) \in \partial_y V(y, f^*(x)) \).
\end{theorem}

\begin{example}[Least Squares Ridge Regression with Kernels] \label{ex:appx_ridge_least_squares}
Using the loss \( V(f(x), y) = (f(x) - y)^2 \), the minimizer \( f^* \) satisfies:
\[
f^*(x) = \sum_{i=1}^N \alpha_i K(x, x_i), \quad \text{where } \boldsymbol{\alpha} = (K + \lambda I)^{-1} \mathbf{y}.
\]
\end{example}

\subsection*{Noteworthy Special Cases}

\subsubsection*{Induced Kernel - General}

Assume to have a well defined real valued kernel $\kappa: \cS \times \cS \to \bR$. 
One can define an operator valued kernel $K : \cS \times \cS \to \cH_{\xi}$ as $K(x,y) := \kappa(x,y)Id_{\cH_{\xi}}$.

\begin{lemma}
    Under the assumptions above $\cH_{K} \simeq \cH_{\kappa} \otimes \cH_{\xi}$.
\end{lemma} 

\begin{proof}
    Recall that $\cH_{K}$ is the closure of $\mathrm{Span}(K(x,\cdot)z ~|~ x \in \cS, z \in \cH_{\xi})$ under the scalar product 
    \begin{align*}
    \sprod{K(x,\cdot)z}{K(y,\cdot)w}_{\cH_{K}} & := \sprod{K(x,y)z}{w}_{\cH_{\xi}} 
    \\ &= \kappa(x,y)\sprod{z}{w}_{\cH_{\xi}} 
    = \sprod{\kappa(x,\cdot)}{\kappa(y,\cdot)}_{\cH_{\kappa}} \sprod{z}{w}_{\cH_{\xi}}
    \end{align*}
    so that $K(x,\cdot)z \mapsto \kappa(x,\cdot) \otimes z \in \cH_{{K}} \otimes \cH_{\xi}$ extends to an isometry.
    
\end{proof}

Under this identification for $x \in \cX$ 
one can prove that
\[
\Phi_x = \int_0^T \kappa(\psi(x)|_{[0,t]}, \cdot) \otimes d \boldsymbol{\xi}(x)_t  \in \cH_{\kappa} \otimes \cH_{\xi}
\]

\subsubsection*{Induced Kernel - Feature Map}
Of special interest is the case where $\cH_{\kappa} \subseteq \bR^N$ and $\cH_{\xi} = \bR^{d_{\xi}}$, obtained when $\kappa(x,y) = \sprod{\mathbb{F}(x)}{\mathbb{F}(y)}_{\bR^N}$ for some feature map $\mathbb{F}: \cS \to \bR^N$.
In this setting $\cH_{\kappa} \otimes \cH_{\xi} \subseteq \bR^{N \times d_{\xi}}$ is a subset of matrices endowed with Hilbert-Schmidt norm and we can write
\begin{equation}
    \Phi_x = \int_0^T 
    \underbrace{
    \mathbb{F}(\psi(x)|_{[0,t]}) 
    }_{N \times 1}
    \underbrace{
    \dD \boldsymbol{\xi}(x)_t^{\top}  
    }_{1 \times d_{\xi}}
    \in \bR^{N \times d_{\xi}}
\end{equation}
and moreover $\mathcal{K}_{\Phi}(x,y) = Tr(\Phi_x\Phi_y^{\top})$.

\vspace{10pt}
This technique is particularly relevant in the context of \emph{Neural Signature Kernels}, a family of path-space kernels that generalize signature kernels and can only be computed through feature maps.

\subsection*{Signatures \& Kernels}

\begin{definition}[Signature Transform] \label{defn:sig_transform}
Let $X:[0,T] \to \R^d \in C^{1-var}([0,T];\R^d)$ be a (piecewise) smooth path. Let us define the simplex $\Delta_T = \{(s,t) : 0 \leq s \leq t \leq T \}$. The signature of $X$ between fixed time $s$ and $t$ is a map
    \begin{align*}
        \text{Sig} : \Delta_T & \to T((\R^d)) \\
        (s,t) & \mapsto \text{Sig}(X_{s,t}) := (1, \text{Sig}^1(X_{s,t}), \dots, \text{Sig}^n(X_{s,t}), \dots )
    \end{align*}
    where the $n$-th order of the signature is defined as
    \begin{align*}
        \text{Sig}^n(X_{s,t}) := \idotsint\displaylimits_{s < u_1 < \dots < u_n < t} dX_{u_1} \otimes \dots \otimes dX_{u_n} \in (\R^d)^{\otimes n}.
    \end{align*}
        The signature is a $T((\R^d))$-valued process, which can be viewed as
        \begin{align*}
        \text{Sig}(X_{0,T}) = \left( \underbrace{
        \begin{matrix}
            \text{ } \vspace{0.2cm}  \\
            \text{Sig}^{\bluebf{\emptyset}}(X_{0,T}) \\
            \text{ } \vspace{0.2cm}
        \end{matrix}}_{\textstyle =1} \text{ } , \text{ }
        \underbrace{\begin{pmatrix}
            \text{Sig}^{\bluebf{1}}(X_{0,T})   \\
            \vdots \\
            \text{Sig}^{\bluebf{d}}(X_{0,T})
        \end{pmatrix}}_{\textstyle \bX_{0,T}^1}
        \text{ } , \text{ }
        \underbrace{\begin{pmatrix}
            \text{Sig}^{\bluebf{11}}(X_{0,T}) & \dots &  \text{Sig}^{\bluebf{1d}}(X_{0,T})   \\
            \vdots & \ddots & \vdots \\
            \text{Sig}^{\bluebf{d1}}(X_{0,T}) & \dots & \text{Sig}^{\bluebf{dd}}(X_{0,T})
        \end{pmatrix}}_{\textstyle \bX_{0,T}^2} \text{ } , \text{ } \dots
        \right),
    \end{align*}
    where each term is given as
    \begin{align*}
        \text{Sig}^{\bluebf{i}_1 \dots \bluebf{i}_n}(X_{s,t}) := \idotsint\displaylimits_{s < u_1 < \dots < u_n < t} dX_{u_1}^{i_1} \otimes \dots \otimes dX_{u_n}^{i_n} \in \R.
    \end{align*}
    The signature of a path can be truncated at any finite order $N \in \N$. We denote the truncated signature up to order $N$ as
    \begin{align*}
        \text{Sig}^{\leq N}(X_{0,T}) : \Delta_T & \to T^{(N)}(\mathbb{R}^d) \\
        (s,t) & \mapsto (1, \text{Sig}^1(X_{s,t}), \dots, \text{Sig}^N(X_{s,t})).
    \end{align*}
\end{definition}

\begin{definition}[Signature Kernel] \label{defn:sig-kernel}
Given two paths $X, X' : [0,T] \to \mathbb{R}^d$, their \emph{signature kernel} is defined by
\[
K_{\mathrm{sig}}(X, X') := \left\langle \mathrm{Sig}(X),\, \mathrm{Sig}(X') \right\rangle,
\]
where the inner product is defined by summing the tensor-level inner products:
\[
\left\langle \mathrm{Sig}(X),\, \mathrm{Sig}(X') \right\rangle = \sum_{m=0}^\infty \left\langle \mathrm{Sig}^{m}(X),\, \mathrm{Sig}^{m} (X')\right\rangle_{(\mathbb{R}^d)^{\otimes m}}.
\]
\end{definition}

\begin{definition}[Randomised Signature Computation] \label{defn:rand_signatures_comp}
    Let \( A_1, \dots, A_d \in \mathbb{R}^{M \times M} \) be random matrices, \( b_1, \dots, b_d \in \mathbb{R}^{M} \) be random shifts, and \( z \in \mathbb{R}^M \) be a random initial state. Let \( \sigma: \mathbb{R}^M \to \mathbb{R}^M \) be a fixed activation function.

    The \emph{Randomised Signature} of \( X \) over \( t \in [0,T] \) is defined as the solution \( Z_t \in \mathbb{R}^M \) to the controlled differential equation:
    \[
    dZ_t = \sum_{i=1}^d \sigma(A_i Z_t + b_i) \, dX_t^i, \quad Z_0 = z \in \mathbb{R}^M.
    \]
    This construction can be viewed as a random projection of the signature, leveraging ideas from the Johnson–Lindenstrauss Lemma \cite{Vempala2004TheMethod}. For more background, we refer the reader to \cite{Cuchiero2021ExpressiveSignature, Compagnoni2022RandomizedData}.
\end{definition}

\begin{definition}[Randomised Signature Kernel] \label{defn:rand_signature_kernel}
Given two (piecewise) smooth paths \( X, Y : [0,T] \to \mathbb{R}^d \) in \( C^{1\text{-}var}([0,T]; \mathbb{R}^d) \), their \emph{Randomised Signature Kernel} is defined as the inner product of their randomized signature embeddings:
\[
K_{\text{r-Sig}}(X, Y) := \langle \text{r-Sig}(X), \text{r-Sig}(Y) \rangle_{\mathbb{R}^M}.
\]
Here, \( \text{r-Sig}(\cdot) \) may be computed 
via the solution of a controlled differential equation (Definition~\ref{defn:rand_signatures_comp}).

This kernel provides a computationally efficient approximation of the full signature kernel, with reduced dimensionality and favourable generalisation properties due to the use of random projections.
\end{definition}

\section{Theorem Proofs}

\subsection{Proof of Theorem~\ref{thm:3-mean_var_solution}} \label{sec:appx_mean_var_proof}

\mvthm*

\begin{proof}
    We start by writing the objective as
    \begin{align*}
        & \ExP^X \left[ \sprod{\phi}{\Phi_X}_{\H_\Phi} - \frac{\eta}{2} \left( \sprod{\phi}{\Phi_X}_{\H_\Phi}^2 - \ExP^X \left[ \sprod{\phi}{\Phi_X}_{\H_\Phi} \right]^2 \right) \right] 
        \\ & =
        \sprod{\phi}{ \ExP^X \left[ \Phi_X \right] }_{\H_\Phi}  - \frac{\eta}{2} \sprod{\phi}{\left( \ExP^X \left[ \Phi_X \otimes \Phi_X \right] - \ExP^X \left[ \Phi_X \right] \otimes \ExP^X \left[ \Phi_X \right]  \right) \phi }_{\H_\Phi}
    \end{align*}
    so that, letting $C_{\mathbb{P}} := \ExP^X \left[ \Phi_X \otimes \Phi_X \right] - \ExP^X \left[ \Phi_X \right] \otimes \ExP^X \left[ \Phi_X \right]$ be the covariance operator, the optimisation becomes over
    \begin{align*} 
    \sprod{\phi}{ \ExP^X \left[ \Phi_X \right] }_{\H_\Phi}  - \frac{\eta}{2} \sprod{\phi}{ C_{\mathbb{P}} \phi }_{\H_\Phi} - \frac{\lambda}{2} \sprod{\phi}{ \phi }_{\H_\Phi}.
    \end{align*}
    
    Define then the map $\iota: L^2_{\mathbb{P}}(\mathcal{X}) \to \H_\Phi$ as
    \begin{equation*}
        \alpha \mapsto \ExP^X[\alpha(X)\Phi_X].
    \end{equation*}
    From \cite[Proposition 4]{Cirone2025RoughHedging} we see that any maximiser,  which exists since $\eta C_{\mathbb{P}} + \lambda \textup{Id}$ is self-adjoint, coercive and by moment assumption \emph{coercive}, must lie in the closure of $\H_{\mathbb{P}} := \iota L^2_{\mathbb{P}}(\mathcal{X})$. 
    In $\H_{\mathbb{P}}$ this problem becomes, noting $\iota 1 = \ExP^X[\Phi_X]$, the optimisation of
    \begin{align*} 
    \sprod{\alpha}{(\iota^\top \circ \iota) 1 }_{L^2_{\mathbb{P}}}  - \frac{\eta}{2} \sprod{\alpha}{ \left( \iota^\top \circ C_{\mathbb{P}} \circ \iota \right) \alpha }_{L^2_{\mathbb{P}}}
    - \frac{\lambda}{2} \sprod{\alpha}{(\iota^\top \circ \iota) \alpha }_{L^2_{\mathbb{P}}}.
    \end{align*}

    Consider $\iota^\top \circ \iota$, by \cite[Lemma 3]{Cirone2025RoughHedging} we see that this map is self-adjoint, that
    \begin{equation*}
        (\iota^\top \circ \iota) \alpha (X) = \Omega\alpha(X) = \ExP^Y \left[ \alpha(Y) \sprod{\Phi_X}{\Phi_Y} \right] = \ExP^Y \left[ \alpha(Y) \K_{\Phi}(X,Y) \right]
    \end{equation*}
    and also 
    \begin{equation*}
        \iota^\top \circ C_{\mathbb{P}} \circ \iota  = \Omega^{\circ 2} - \Omega1 \otimes \Omega1 = \Omega \circ (\textup{Id} - 1 \otimes 1) \circ \Omega.
    \end{equation*}  
    Hence, we can write the objective as
    \begin{align*} 
    \sprod{\alpha}{\Omega 1 }_{L^2_{\mathbb{P}}}  - \frac{1}{2} \sprod{\alpha}{ \Omega \circ \left(\eta \, (\textup{Id} - 1 \otimes 1) \circ \Omega+ \lambda \, \textup{Id}\right) \alpha }_{L^2_{\mathbb{P}}}.
    \end{align*}

    The assumptions on $\Omega$ guarantee that the solution to the optimisation is given by
    \begin{equation*}
        \alpha^* = \left[ \Omega \circ \left(\eta \, (\textup{Id} - 1 \otimes 1) \circ \Omega+ \lambda \, \textup{Id}\right)\right]^{-1} \Omega 1
        = \left(\eta \, (\textup{Id} - 1 \otimes 1) \circ \Omega + \lambda \, \textup{Id}\right)^{-1} 1
    \end{equation*}
    and the conclusion follows since 
    \[
    \Xi_{\mathbb{P}}^\Phi =  (\textup{Id} - 1 \otimes 1) \circ \Omega \,.
    \]
\end{proof}

\subsection{Proof of Theorem~\ref{thm:3-spectral_alpha}} \label{sec:spectral_alpha_proof}

\alphathm*

\begin{proof}
We start by rewriting the system
\begin{align*}
    \left( \lambda \text{Id}_N + \eta \Xi^\Phi_{\P} \right)(\alpha^*) = \boldsymbol{1}_N,    
\end{align*}
and observe that \(\boldsymbol{1}_N = \sqrt{N} \mathbf{e}_N\), and that
\[
\Xi^\Phi_{\P} = \frac{1}{N} \K_{\Phi}(\bfX, \bfX)\left( \textup{Id}_N - \mathbf{e}_N \mathbf{e}_N^\top \right).
\]
We can define the unit-norm vector $\mathbf{e}_N := \sqrt{N}^{-1} (1,1,\dots,1)^{\top}$, then we have $\frac{1}{N} \boldsymbol{1}_N \boldsymbol{1}_N^\top =\mathbf{e}_N \mathbf{e}_N^\top$. Now, defining
\[
A := \lambda \textup{Id}_N + \frac{\eta}{N} \K_{\Phi}(\bfX, \bfX),
\]
we can express the LHS of the system as
\begin{align*}
\left(\lambda \textup{Id}_N + \eta \Xi^\Phi_{\P} \right) & = \lambda \textup{Id}_N + \frac{\eta}{N} \K_{\Phi}(\bfX, \bfX) \left(\textup{Id}_N - \boldsymbol{e}_N \boldsymbol{e}_N^\top \right) \\
&  = \lambda \textup{Id}_N + \frac{\eta}{N} \K_{\Phi}(\bfX, \bfX) \left(\textup{Id}_N - \boldsymbol{e}_N \boldsymbol{e}_N^\top \right) + \lambda \boldsymbol{e}_N \boldsymbol{e}_N^\top - \lambda \boldsymbol{e}_N \boldsymbol{e}_N^\top \\
& = \lambda \textup{Id}_N (\textup{Id} - \boldsymbol{e}_N \boldsymbol{e}_N^\top) + \frac{\eta}{N} \K_{\Phi}(\bfX, \bfX) \left(\textup{Id}_N - \boldsymbol{e}_N \boldsymbol{e}_N^\top \right) + \lambda \boldsymbol{e}_N \boldsymbol{e}_N^\top \\
& = \left( \lambda \textup{Id}_N + \frac{\eta}{N} \K_{\Phi}(\bfX, \bfX) \right)(\textup{Id}_N - \boldsymbol{e}_N \boldsymbol{e}_N^\top) + \lambda \boldsymbol{e}_N \boldsymbol{e}_N^\top \\
& = A(\textup{Id}_N - \boldsymbol{e}_N \boldsymbol{e}_N^\top) + \lambda \boldsymbol{e}_N \boldsymbol{e}_N^\top.
\end{align*}
Hence, we are able to express the system as 
\begin{align*}
    \left(\lambda \mathbf{e}_N \mathbf{e}_N^\top + A (\textup{Id}_N - \mathbf{e}_N \mathbf{e}_N^\top) \right) (\alpha^*) = \sqrt{N} \mathbf{e}_N.
\end{align*}
Applying the rank-one inverse identity:
\[
\left[ \lambda \mathbf{v}\mathbf{v}^{\top} + A(\textup{Id}_N - \mathbf{v}\mathbf{v}^{\top}) \right ]^{-1} \mathbf{v} = \frac{A^{-1} \mathbf{v}}{\lambda \mathbf{v}^{\top} A^{-1} \mathbf{v}},
\]
with \(\mathbf{v} = \mathbf{e}_N\), we obtain
\[
\alpha^* = \frac{ \sqrt{N} A^{-1} \mathbf{e}_N }{ \lambda \mathbf{e}_N^\top A^{-1} \mathbf{e}_N }.
\]
If \(A = \sum_{k=1}^N \gamma_k \mathbf{u}_k \mathbf{u}_k^\top\), then
\[
A^{-1} \mathbf{e}_N = \sum_{k=1}^N \frac{1}{\gamma_k} (\mathbf{u}_k^\top \mathbf{e}_N) \mathbf{u}_k,
\quad \text{and} \quad
\mathbf{e}_N^\top A^{-1} \mathbf{e}_N = \sum_{k=1}^N \frac{(\mathbf{u}_k^\top \mathbf{e}_N)^2}{\gamma_k},
\]
which yields the final expression.
\end{proof}

\section{Mean-Variance with Stochastic Drift} \label{sec:appx-solution-variance}

\subsection{The Underlying Dynamics}

In Section~\ref{subsec:stoch_drift_markowitz}, we consider an example in which we have a multi-dimensional underlying asset process that is composed of a stochastic drift component and a martingale component with static covariance structure. That is, we consider the $d$-dimensional asset process $X = (X^1,\dots,X^d)$, where
\begin{align*}
    dX_t = \mu_t dt + dM_t
\end{align*}
where $M$ is a continuous martingale such that
\begin{align*}
d[M^m, M^n]_t = \Sigma_{mn} dt, \quad \text{for } m, n = 1, \dots, d,
\end{align*}
where \( \Sigma\in \mathbb{R}^{d \times d} \) is a symmetric nonnegative definite covariance matrix. 
\subsection{The Objective}
Suppose we wish to maximise the mean-variance objective with terminal variance penalty,
\begin{align*}
   J(\xi) = \Ex\left[ \int^T_0 \xi_t dX_t \right] - \frac{\eta}{2}\text{Var}\left[ \int^T_0 \xi_t dX_t \right].
\end{align*}
Using the dynamics of $X$, we can reconstruct the objective as,
\begin{align*}
     \notag J(\xi) & = \Ex \left[ \int^T_0 \xi_t dX_t \right] - \frac{\eta}{2} \text{Var} \left[ \int^T_0 \xi_t dX_t \right] \\
    \notag & = \Ex \left[ \int^T_0 \xi_t^\top \mu_t dt\right] - \frac{\eta}{2} \text{Var} \left[ \int^T_0 \xi_t^\top \mu_t dt +\int^T_0 \xi_t^\top dM_t  \right] \\
     & =  \underbrace{\Ex \left[ \int^T_0 \xi_t^\top \mu_t dt\right]}_{(1)} - \frac{\eta}{2}\underbrace{ \Ex \left[ \int_0^T \xi_t^\top \Sigma \xi_t  dt \right]}_{(2)} - \frac{\eta}{2}\underbrace{ \text{Var} \left[ \int^T_0 \xi_t^\top \mu_t dt \right]}_{(3)} \\
     & \quad \quad - \eta\underbrace{ \text{Cov}\left(\int^T_0 \xi_t^\top \mu_t dt, \int^T_0 \xi_t^\top dM_t \right)}_{(4)}.
\end{align*}
We can observe the following:
\begin{itemize}
\vspace{-0.2cm}
\itemsep-0.5em
    \item Term (3) is zero if and only if $\mu_t$ and $\xi_t$ are deterministic, which they are not in our setting.
    \item Term (4) vanishes if $\mu_t$ is independent of the martingale component $M$. For instance if $\mu_t$ is adapted to a filtration $\F_t$ independent of $M_t$. This would be the case, for example, if $\mu_t$ were an exogenous signal (e.g. from an inhomogeneous OU process independent of $X$). However, in the standard OU process setting where $\mu_t = \kappa(\theta - X_t)$, the drift depends on $X_t$, which is itself driven by $M$, introducing correlation between $\mu_t dt$ and $dM_t$.
\end{itemize}

\subsection{The System to Solve}
Now, in order to maximise $J$ with respect to $\xi$, we can compute the Gateau derivative of $J$ as 
\begin{align*}
    \langle J'(\xi), v \rangle = \lim_{\varepsilon \to 0} \frac{J(\xi + \varepsilon v) - J(\xi)}{\varepsilon}.
\end{align*}
To obtain $J(\xi + \epsilon v)$, we first expand (1), obtaining
\begin{align*}
    \Ex \left[ \int^T_0 \xi_t^\top \mu_t dt \right] + \varepsilon \Ex \left[ \int^T_0 v_t^\top \mu_t dt \right],
\end{align*}
expanding (2), we obtain
\begin{align*}
    \Ex \left[ \int_0^T \xi_t^\top \Sigma \xi_t dt \right] + 2\varepsilon \Ex \left[ \int_0^T v_t^\top \Sigma \xi_t dt \right] + \mathcal{O}(\varepsilon^2),
\end{align*}
for (3), we have
\begin{align*}
    & \text{Var} \left(\int^T_0 \xi_t^\top \mu_t dt\right) + 2\varepsilon \text{Cov} \left( \int_0^T \xi_t^\top \mu_t dt, \int_0^T  v_s^\top \mu_s  ds \right) + \mathcal{O}(\varepsilon^2) 
\end{align*}
where the covariance term can be expressed as
\begin{align*}
    \text{Cov} \left( \int_0^T \xi_s^\top \mu_s \, ds,\; \int_0^T  v_t^\top \mu_t \, dt \right)
    & = \Ex \left[ \int_0^T \int_0^T (\xi_s^\top \mu_s)(v_t^\top \mu_t) \, dt \, ds \right]
     \\
     & \quad \, \, \, - \Ex \left[ \int_0^T \xi_s^\top \mu_s \, ds \right] \Ex\left[ \int_0^T  v_t^\top \mu_t \, dt \right] \\
    & = \Ex \left[ \int_0^T v_t^\top \left( \mu_t \int_0^T \xi_s^\top \mu_s \, ds \right) dt \right] \\
     & \quad \, \, \, - \Ex \left[ \int_0^T v_t^\top \left( \mu_t \Ex \left[ \int_0^T \xi_s^\top \mu_s \, ds \right] \right) dt \right] \\
    & = \Ex \left[ \int_0^T v_t^\top \left( \mu_t \int_0^T \left( \xi_s^\top \mu_s - \Ex[\xi_s^\top \mu_s] \right) ds \right) dt \right].
\end{align*}
using Fubini. Finally, for (4), we obtain
\begin{align*}
    \text{Cov} \left( \int_0^T \xi_t^\top \mu_t dt,\ \int_0^T \xi_t^\top dM_t \right)
	+ \varepsilon \,\text{Cov}\left( \int_0^T v_t^\top \mu_t dt,\ \int_0^T \xi_t^\top dM_t \right) \\
	+ \varepsilon\, \text{Cov} \left( \int_0^T \xi_t^\top \mu_t dt,\ \int_0^T v_t^\top dM_t \right)
	+ \mathcal{O}(\varepsilon^2)
\end{align*}
where we can rewrote the first order epsilon terms as 
\begin{align*}
    \text{Cov}\left( \int_0^T v_t^\top \mu_t dt,\ \int_0^T \xi_t^\top dM_t \right) & = \Ex \left[ \int_0^T \xi_s^\top \left( \int_0^T v_t^\top \mu_t \, dt \right) dM_s \right]  \\
     \text{Cov} \left( \int_0^T \xi_t^\top \mu_t dt,\ \int_0^T v_t^\top dM_t \right) & = \Ex \left[ \int_0^T v_s^\top \left( \int_0^T \xi_t^\top \mu_t \, dt \right) dM_s \right] 
\end{align*}

Now, collecting all first-order $\varepsilon$-terms,
\begin{align*}
    \langle J'(\xi), v \rangle & = \Ex \left[ \int^T_0 v_t^\top \left( \mu_t - \eta \Sigma \xi_t \right) dt \right] - \frac{\eta}{2}
    \Ex \left[ \int_0^T v_t^\top \left( \mu_t \int_0^T \left( \xi_s^\top \mu_s - \Ex[\xi_s^\top \mu_s] \right) ds \right) dt \right] \\
    & \quad \, \, - \eta \Ex \left[ \int_0^T \left[ v_s^\top \left( \int_0^T \xi_t^\top \mu_t \, dt \right) + \xi_s^\top \left( \int_0^T v_t^\top \mu_t \, dt \right) \right]  dM_s \right] 
\end{align*}
We require this to be equal to zero for all perturbations $v_t$, conditional on the filtration $\F_t$.
Now, since $v \in L^2([0,T]; \mathbb{R}^d)$ is an arbitrary perturbation, we are able to express this first order condition as
\begin{align*}
    \Ex \left[ \int^T_0 v_t^\top \left( \mu_t - \eta \Sigma \xi_t - \frac{\eta}{2} \mu_t \int_0^T \left( \xi_s^\top \mu_s - \Ex[\xi_s^\top \mu_s] \right) ds\right) dt \right] & = \Ex \left[ \int_0^T v_s^\top \left( \int_0^T \xi_t^\top \mu_t \, dt \right) dM_s \right] \\
    & \ \ + \Ex \left[ \int_0^T \xi_s^\top \left( \int_0^T v_t^\top \mu_t \, dt \right) dM_s \right].
\end{align*}
Unfortunately this is not straightforward to solve due to the additional martingale term on the right hand side. If, term (4) were to be equal to zero (in the case of independent $\mu, M$) then we would be required to solve for
the pointwise condition
\begin{align*}
    \Ex_t[\mu_t] - \eta \Sigma \xi_t - \frac{\eta}{2} \Ex_t[\mu_t] \int_0^T \left( \Ex_t[\xi_s^\top \mu_s] - \Ex[\xi_s^\top \mu_s] \right) ds = 0 \quad \forall t\in[0,T],
\end{align*}
which is an integral equation for $\xi_t$.
We distinguish that the conditional expectation at time $t$ is denoted using the subscript $t$ as $\Ex_t[\cdot] = \Ex[\cdot \, | \, \F_t]$, where as the conditional expectation includes no subscript.
Clearly, we can see that if the remainder covariance term is zero, then we recover the classical solution of Markowitz in \eqref{eq:3-markowitz_optimal}, however this is only true if the drift or our trading strategy is deterministic which it is not. The optimal solution therefore in this case depends on maximising the expected return through the drift $\mu_t$, but also requires to anticipate future moves in the drift term $\mu_s, s\in[t,T]$ (i.e its decay), as well as future changes in its position, in order to reduce the variance of the terminal PnL. 

\end{appendices}
\printbibliography[title={References}]

\end{document}